\documentclass{article}
\usepackage{graphicx} 
\usepackage[utf8]{inputenc}
\usepackage{natbib}
\usepackage{authblk}
\usepackage{xcolor}
\usepackage{url}
\usepackage{bm}
\usepackage{alphabeta}
\usepackage{amssymb}
\usepackage{amsmath,bm,calc,float}
\usepackage{amsthm}
\usepackage{hyperref}
\usepackage[english]{babel}
\usepackage{amsthm}
\usepackage{comment}
\newtheorem{definition}{Definition}[section]
\newtheorem{theorem}{Theorem}[section]

\usepackage{subfigure}
\usepackage[english]{babel}
\usepackage{graphicx}
\hypersetup{
    colorlinks=true,
    linkcolor=blue,
    citecolor=blue
    } 
\usepackage[hmarginratio=1:1]{geometry}
\usepackage{cleveref}
\usepackage{marvosym}
\usepackage{array, makecell} 
\theoremstyle{remark}
\newtheorem{case}{Case}

\title{A Bivariate DAR($1$) model for ordinal time series}
\date{}

\author{Anna Nalpantidi
  and  Dimitris Karlis}

\affil{ Department of Statistics, \\ Athens University of Economics and Business,  Greece}

\begin{document}

\maketitle

\begin{abstract}
We present a bivariate vector valued discrete autoregressive model of order $1$ (BDAR($1$)) for discrete time series. The BDAR($1$) model assumes that each time series follows its own univariate DAR($1$) model with dependent random mechanisms that determine from which component the current status occurs and dependent innovations. The joint distribution of the random mechanisms which are expressed by Bernoulli vectors are proposed to be defined through copulas. The same holds for the joint distribution of innovation terms. Properties of the model are provided, while special focus is given to the case of bivariate ordinal time series. A simulation study is presented, indicating that model provides efficient estimates even in case of moderate sample size. Finally, a real data application on unemployment state of two countries is presented, for illustrating the proposed model. 
\end{abstract}


\section{Introduction} 

An ordinal time series is a sequence of observations $(Z_t)_{t\in \mathcal{Z}}$ with
$Z_t \in \mathcal{S}=\{s_1,s_2,\ldots,s_d\}$, where $s_1<s_2<\ldots<s_d$,
that evolves over time. This type of data is common across various fields, including environmental studies (\cite{gottlein1992ordinal}, \cite{liu2022modeling2}, \cite{liu2022modeling1}, \cite{jahn2024nonlinear}), sports \citep{fokianos2003regression}, healthcare \citep{fokianos2003regression}, and economics \citep{weiss2019distance}. For instance, \cite{liu2022modeling2} examined air quality in three major Chinese cities with the aim of managing air pollution. For each city, the daily air quality index was recorded on a six-level ordinal scale: (1) excellent, (2) good, (3) slightly polluted, (4) moderately polluted, (5) heavily polluted, and (6) severely polluted. Clearly, there is a natural ordering among these categories. In another application, \cite{weiss2019distance} studied credit ratings—an assessment of the ability of a debtor (e.g., individual, government, or corporation) to repay debt—for EU countries, based on Standard \& Poor (S\&P) ratings. The dataset covers the period from January 2000 to December 2017, with each country assigned one of 23 possible ratings ranging from "D" (worst) to "AAA" (best). Intermediate ratings between “CCC” and “AA” are further refined using plus or minus signs.

In some cases, even when continuous values are recorded, they may be considered unreliable or imprecise. In such situations, transforming them into ordinal categories can improve robustness, as in the case of EEG data \citep{keller2007ordinal}.

This type of time series requires special methodological care, since models must respect both the discrete nature of the data and the natural ordering of states. In a series of works \cite{jacobs1978discrete}, \cite{jacobs1978discrete2}, and \cite{jacobs1978discrete3} introduced discrete versions of ARMA models (DARMA), including a special case of the AR model (DAR), which accounts only for non-negative dependence. Later, \cite{jacobs1983stationary} proposed the new DARMA model (NDARMA), offering a simpler and more intuitive dependence structure, closer in spirit to the familiar ARMA model.

The main idea of the NDARMA($p,q$) model is that each current value is generated as a random choice among the $p$ past values of the series, the $q$ past innovation terms, and a new innovation variable. This random mixture is typically represented by a multinomial vector. The innovation process itself is assumed to be discrete with the same state space $\mathcal{S}$, ensuring that the time series always takes values within the correct support.

NDARMA models are appropriate for any type of discrete valued time series namely count, ordinal, nominal and binary as a special case. Following the approach of \cite{weiss2008measuring}, we can  define the NDARMA(p,q) model as follows:

\begin{definition}\label{de1}
{\bf The NDARMA($p,q$) model:} 
Let the observations $(Z_t)_{t\in \mathcal{Z}}$ and innovations $(\epsilon_t)_{t\in\mathcal{Z}}$ be discrete valued processes with state space $\mathcal{S}=\{s_1,\ldots,s_d\}$. Innovation process $(\epsilon_t)_{t\in \mathcal{Z}}$ is i.i.d. with marginal distribution $\boldsymbol{p}=(p_{s1},\ldots,p_{s_d})$, where $p_{s_i}=P(\epsilon_t=s_i)$ and they are assumed independent of $(Z_s)_{s<t}$. To obtain the random mechanism which chooses the current value of the process $Z_t$, we consider the i.i.d. multinomial vectors
\begin{equation*}
(\alpha_{1t},\ldots,\alpha_{pt},\beta_{0t},\beta_{1t},\ldots,\beta_{qt})\sim \text{Multinomial}(1;\phi_1,\ldots,\phi_p,\psi_0,\ldots,\psi_q)
\end{equation*}
which are independent of $(\epsilon_t)_{t\in \mathcal{Z}}$ and $(Z_s)_{s<t}$. Then, $(Z_t)_{t\in \mathcal{Z}}$ is said to be a NDARMA(p,q) process, if it follows the recursion: 
\begin{equation}
  Z_t=\alpha_{1t} Z_{t-1}+\ldots+ \alpha_{pt} Z_{t-p}+\beta_{0t} \epsilon_t+\beta_{1t} \epsilon_{t-1} +\ldots+\beta_{qt} \epsilon_{t-q}. 
\end{equation}
\end{definition}
For $q=0$ we have the special case of DAR(p) model and for $p=0$ we have the special case of DMA(q) model. 

The idea of using random mixture  to adapt ARMA dependence structure in discrete time series can also be achieved through Pegram's operator. More specifically, \cite{pegram1980autoregressive}
proposed a class of discrete-valued AR(p) processes. Pegram’s operator $*$ is a mixing operator which mixes two or more random variables. 

\begin{definition}\label{de2}
{\bf Pegram's operator: } 
For a series of $m$ independent discrete random variables $U_i$, $i=1,\ldots,m$, 
Pegram's operator mixes the $U_i$ 
with probabilities $\phi_i$ with $\sum \phi_i =1$, denoted as
\begin{equation*}
Z= (\phi_1,U_1)*\ldots*(\phi_m, U_m),
\end{equation*}
with the corresponding marginal probability function to be
\begin{equation*}
P(Z=j) = \sum\limits_{i=1}^m \phi_i P(U_i=j).  
\end{equation*}
\end{definition}

Based on this mixing operator, \cite{pegram1980autoregressive}
created a class of stationary AR($p$) process, 
denoted as PAR($p$) process. It has been shown that the PAR($p$) process is 
equivalent to the DAR(p) process of \cite{jacobs1978discrete} defined in a similar manner. \cite{biswas2009discrete} extended this PAR($p$) process to PARMA(p, q) 
introducing moving average (MA) terms. This model is equivalent to the
NDARMA process of \cite{jacobs1983stationary}.
See also \cite{angers2017bayesian} for an application of such models.
Pegram's operator has been also used to define count data time series in 
\cite{khoo2017modelling}.

In the context of times  series $Z_t$ an autoregressive model of the form
\begin{equation*}
Z_t = (\phi,Z_{t-1})*(1-\phi,\epsilon_t) ,
\end{equation*}
implies that the current value at time $t$ is
either that of the previous observation at time $t-1$, with probability $\phi$ or a new one coming from some innovation random variable $\epsilon_t$, with probability $1-\phi$. In fact the above defines an autoregressive of order 1 model. Pegram's operator can be used to define  a very rich family of time series models. 

Extensions of NDARMA models have been proposed in various directions. In particular, \cite{moller2020generalized} introduced the Generalized DARMA (GDARMA) model, which addresses a key limitation of the standard NDARMA framework. Specifically, the random mechanism underlying NDARMA models often produces long runs of repeated values, making them unsuitable for many practical applications. The GDARMA model overcomes this issue by incorporating data-specific variation operations. This innovation not only resolves the problem of long runs but also enhances the model’s flexibility, allowing it to accommodate a wide range of quantitative time series.

Moreover, for the special case of binary time series, \cite{jentsch2019generalized} proposed the Generalized Binary Autoregressive (gbAR) model, which is capable of capturing both positive and negative dependence. This is possible in the binary setting, since negative dependence can naturally be interpreted as a tendency toward the opposite state. The authors later extended the framework to the multivariate case \citep{jentsch2022generalized}, incorporating cross-correlation terms across series as well as covariance among the current innovation terms.

On the other hand, the literature on bivariate models for ordinal time series is relatively underdeveloped. Despite the growing interest in modeling multiple outcomes, the lack of well-established and easily applicable multivariate distributions for ordinal data presents a significant challenge. Moreover, two types of dependence must be accounted for: serial correlation within each series and cross-correlation between the two series.

In the context of bivariate longitudinal ordinal data, \cite{todem2007latent} proposed a latent-variable framework. Their approach assumes that each observed ordinal outcome is driven by an underlying continuous latent variable, which is modeled using a linear mixed model. Serial correlation within each outcome is captured through random effects, while cross-correlation is modeled by assuming a joint Gaussian distribution for both the error terms and the random effects of the two latent variables. In a different approach, \cite{lee2013flexible} proposed modeling bivariate ordinal longitudinal data using marginalized models. Specifically, the marginal mean of each outcome is linked to a set of covariates and estimated via a cumulative logit model. For the correlation structure, a cumulative logit model with random effects is employed. To capture cross-correlation between the two outcomes, both within and across time points, the authors introduced a Kronecker product structure for the covariance matrix of the random effects. From a different perspective, \cite{nikoloulopoulos2019coupling} presented a copula-based bivariate panel ordinal model. 

The contribution of the present paper is as follows. Addressing the lack of bivariate models for ordinal time series, we introduce a bivariate DAR($1$) model for this setting. In our approach, each time series is modeled as a DAR($1$) process with a specific random mixture vector that captures non-negative serial dependence. The correlation between the two series is represented through two components. First, although each series follows a distinct random mechanism, the mechanisms themselves are assumed to be correlated, requiring the specification of a joint distribution for the two Bernoulli random variables. Second, we allow the innovation terms of the two series to be correlated, which requires defining a joint distribution for multinomial random vectors. Since specifying these bivariate distributions directly is challenging, we propose the use of copula functions, which provide a flexible framework to model joint distributions and select an appropriate dependence structure based on the data. This mechanism is expected to effectively capture the cross-dependence between the two series. Importantly, we avoid introducing cross-correlation terms within each DAR($1$) model, as differences in the state spaces of the two series could lead to inconsistent or invalid values.

The paper is organized as follows. In \hyperref[Sec2]{Section 2}, we review the DAR($1$) model for univariate discrete time series and extend it to the bivariate case, introducing the BDAR($1$) model for discrete, and in particular, ordinal time series. The properties of the BDAR($1$) model are presented in \hyperref[Sec3]{Section 3}. In \hyperref[Sec4]{Section 4}, we discuss the estimation procedure, while \hyperref[Sec5]{Section 5} presents a simulation study to evaluate the model’s performance under varying sample sizes. Based on the proposed model, \hyperref[Sec6]{Section 6} provides a joint analysis of the unemployment rates in Slovakia and the Czech Republic. Finally, \hyperref[Sec7]{Section 7} summarizes the main findings and outlines potential directions for future research.

\section{Definition of BDAR (1) model}
\label{Sec2}
\subsection{DAR($1$) model}
According to \cite{weiss2008measuring}, DAR($1$) is defined as follows:

\begin{definition}\label{de3}
{\bf The DAR($1$) model:} 
Let observations $(Z_t)_{t\in \mathcal{Z}}$ and the innovation terms $\epsilon_{t}$ be discrete processes with state space $\mathcal{S}$. Innovation $\epsilon_{t}$ is i.i.d. with marginal distribution $\boldsymbol{p}_{\epsilon}$,$(P(\epsilon_t=i)=p_{\epsilon_i}, i\in \mathcal{S})$ and they are assumed to be independent of $(Z_s)_{s<t}$. The random mixture is obtained through the i.i.d. Bernoulli random variables:
\begin{equation}
  \alpha_t \sim Bernoulli(\phi),\hspace{1cm} \mbox{for}~ t\in \mathcal{Z}.  
\end{equation}
Then $(Z_t)_{t\in\mathcal{Z}}$ is said to be a DAR($1$) process if it follows the recursion:
\begin{equation}
  Z_t=\alpha_t Z_{t-1}+(1-\alpha_t)\epsilon_t. 
\label{Eq:DAR($1$)}
\end{equation}
where $\alpha_t$ is independent to $\epsilon_t$ and $(Z_s)_{s<t}$.
\end{definition}

The model in \ref{Eq:DAR($1$)} implies that $Z_t$ will be $Z_{t-1}$ with probability $\phi$ or $\epsilon_t$ with probability $1-\phi$. DAR($1$) model is a Markov chain of order $1$. It is a stationary process for $0\leq\phi<1$ with marginal distribution the same as innovation term $\epsilon_t$. It also has the correlation structure of an AR(1) model, while we can derive the set of Yule-Walker equation, to notice that all correlations are always greater or equal to zero.

\subsection{BDAR($1$) model for discrete time series}

At this point we extend the DAR($1$) model to the bivariate setting, introducing the BDAR($1$) model. Consider two discrete-valued time series, which may be of different types (count, ordinal, or binary). The two series evolve jointly over time, with each process influencing the other. Individually, each series follows a DAR($1$) process. To capture their cross-dependence, we assume that the random mixtures governing the two processes are associated, and we model their joint distribution through a suitable copula function. In addition, we allow the innovation terms of the two series to be correlated, again using a copula function to specify their joint distribution. In this framework, the state of each time series depends not only on its own past but also on the evolution of the other process. This interaction arises because both the Bernoulli random variables driving the random mixtures and the innovation terms are correlated across the two series.

The idea of using copulas is based on their advantage to allow defining multivariate distributions easily \citep{nelsen2006introduction}. In addition, they are also flexible in the way that allow for a great variety of dependence structure, choosing the appropriate copula. An important note is that, in case of continuous margins the copula is unique. However, according to Sklar's theorem, in case of discrete margins the copula is not unique. However, the distribution is still valid and plenty of examples using copula to define  multivariate distributions for discrete variables can be found in the literature \citep{nikoloulopoulos2008multivariate, panagiotelis2012pair,nikoloulopoulos2019coupling}. 

\begin{definition}\label{de4}
{\bf The BDAR($1$) model:} 
Let $(\boldsymbol{Z}_t)_{t\in\mathcal{Z}}=(Z_{1t},Z_{2t})^{'}_{t\in\mathcal{Z}}$ be a $2$-dimensional observed discrete time series with state space $\mathcal{S}_1 \times \mathcal{S}_2$ and $(\boldsymbol{\epsilon}_t)_{t \in \mathcal{Z}}=(\epsilon_{1t},\epsilon_{2t})^{'}_{t\in\mathcal{Z}}$ 
be an i.i.d. $2$-dimensional discrete innovation processes, with mean value $\boldsymbol{\mu}_\epsilon=(\mu_{\epsilon_1},\mu_{\epsilon_2})^{'}$ and covariance matrix $\Sigma_\epsilon$, such that $\boldsymbol{\epsilon}_t$ is independent of $(\boldsymbol{Z}_s)_{s<t}$. In addition, we assume that $\epsilon_{1t}$ and $\epsilon_{2t}$ have marginal distributions $\boldsymbol{p_\epsilon^{(1)}}, (P(\epsilon_{1t}=i)=p^{(1)}_{\epsilon_i}, i\in\mathcal{S}_1)$ and $\boldsymbol{p_\epsilon^{(2)}}, (P(\epsilon_{2t}=i)=p^{(2)}_{\epsilon_i},i\in \mathcal{S}_2)$, respectively, and  cdfs denoted as $F_{\epsilon_{1t}}(\cdot)$  and
$F_{\epsilon_{2t}}(\cdot)$. 
Their joint cumulative distribution function is given by a copula function $C(\cdot,\cdot)$ with dependence parameter $\delta_{\epsilon}$. Then, the joint probability mass function (pmf) is given using:
\begin{eqnarray*}
P(\epsilon_{1t},\epsilon_{2t})&=&C(F_{\epsilon_{1t}}(\epsilon_{1t}),F_{\epsilon_{2t}}(\epsilon_{2t});\delta_{\epsilon})-C(F_{\epsilon_{1t}}(\epsilon_{1t}-1),F_{\epsilon_{2t}}(\epsilon_{2t});\delta_{\epsilon})\\
    &-&C(F_{\epsilon_{1t}}(\epsilon_{1t}),F_{\epsilon_{2t}}(\epsilon_{2t}-1);\delta_{\epsilon})+
    C(F_{\epsilon_{1t}}(\epsilon_{1t}-1),F_{\epsilon_{2t}}(\epsilon_{2t}-1);\delta_{\epsilon}).
\end{eqnarray*}
We denote the joint pmf as $p_{\epsilon_{ij}}=P(\epsilon_{1t}=i,\epsilon_{2t}=j)$. 
To obtain the random mechanism which randomly selects between past value $\boldsymbol{Z}_{t-1}$ and the innovation $\boldsymbol{\epsilon}_t$, we consider i.i.d. Bernoulli random variables. The random mechanism is assumed to be separate for each time series. We assume different Bernoulli random variables for each time series 
but the Bernoulli random variables are assumed to be dependent. Their joint cumulative distribution is also defined through a copula function $C(\cdot,\cdot)$ with dependent parameter $\delta_\alpha$. 
\begin{eqnarray*}
\alpha_{1t} &\sim& Bernoulli(\phi_1)\\ 
\alpha_{2t} &\sim& Bernoulli(\phi_2)
\end{eqnarray*}
with joint probability mass function
\begin{eqnarray*}
P(\alpha_{1t},\alpha_{2t})&=&C(F_{\alpha_{1t}}(\alpha_{1t}),F_{\alpha_{2t}}(\alpha_{2t});\delta_{\alpha})-C(F_{\alpha_{1t}}(\alpha_{1t}-1),F_{\alpha_{2t}}(\alpha_{2t});\delta_{\alpha})\\
    &-&C(F_{\alpha_{1t}}(\alpha_{1t}),F_{\alpha_{2t}}(\alpha_{2t}-1);\delta_{\alpha})+
    C(F_{\alpha_{1t}}(\alpha_{1t}-1),F_{\alpha_{2t}}(\alpha_{2t}-1);\delta_{\alpha}),
\end{eqnarray*}
where $F_{\alpha_{1t}}(\cdot)$ and
$F_{\alpha_{2t}}(\cdot)$
are the cdfs of the Bernoulli random variables. 
We denote $\pi_{ij}=P(a_{1t}=i,a_{2t}=j)$. Then, we have four possible outcomes with probabilities:
\begin{align*}
&\pi_{11}=P(\alpha_{1t}=1,\alpha_{2t}=1), \quad \text{then} \quad \boldsymbol{Z}_t=\boldsymbol{Z}_{t-1},\\
&\pi_{10}=P(\alpha_{1t}=1,\alpha_{2t}=0), \quad \text{then} \quad (Z_{1t},Z_{2t})^{'}=(Z_{1,t-1},\epsilon_{2t})^{'},\\
&\pi_{01}=P(a_{1t}=0,\alpha_{2t}=1), \quad \text{then} \quad (Z_{1t},Z_{2t})^{'}=(\epsilon_{1t},Z_{2,t-1})^{'},\\
&\pi_{00}=P(\alpha_{1t}=0,\alpha_{2t}=0), \quad \text{then} \quad \boldsymbol{Z}_t=\boldsymbol{\epsilon}_{t}.
\end{align*}
$\boldsymbol{\alpha}_{t}$ are assumed to be independent of $\boldsymbol{\epsilon}_t$ and of $(\boldsymbol{Z}_s)_{s<t}$. 

$\boldsymbol{Z}_t$ is said to be a BDAR($1$) process if it follows the recursion:
\begin{align}
\begin{bmatrix} Z_{1t} \\ Z_{2t} \end{bmatrix} 
 &=\begin{bmatrix} \alpha_{1t} \\ \alpha_{2t} \end{bmatrix} \odot \begin{bmatrix} Z_{1,t-1} \\ Z_{2,t-1} \end{bmatrix}+
\begin{bmatrix} 1-\alpha_{1t} \\ 1-\alpha_{2t} \end{bmatrix} \odot \begin{bmatrix} \epsilon_{1t} \\ \epsilon_{2t} \end{bmatrix} \nonumber \\
&\text{or equivalent} \nonumber \\
\boldsymbol{Z}_t&=\boldsymbol{\alpha}_t \odot \boldsymbol{Z}_{t-1} + \boldsymbol{\beta}_t \odot \boldsymbol{\epsilon_t}
\label{Eq:BDAR($1$)}
\end{align}
where $\boldsymbol{\alpha}_t=(\alpha_{1t},\alpha_{2t})^{'}$, $\boldsymbol{\beta}_t=1-\boldsymbol{\alpha}_t$ and $\odot$ the Hadamard (element wise) product.
\end{definition}

\begin{case}\label{case1}
Although our primary focus is on ordinal time series, the model described in \ref{Eq:BDAR($1$)} can be applied more generally to the joint modeling of two discrete-valued time series, whether count, ordinal, or binary. This means that the methodology remains valid even in the case of mixed-type series. We exclude, however, the case of nominal time series. The reason is that the joint distribution of the innovations is defined via a copula function, which requires specification of the marginal cumulative distribution functions. Such specification suppose in advance an ordering of the possible states. For nominal time series, no natural ordering exists, and imposing an arbitrary order could distort the results.
\end{case}

\begin{case}\label{case2}
We assume that all dependence is captured through the joint distribution of the innovations and the joint distribution of the random mixtures. By doing so, we avoid the use of explicit cross-correlation terms, which allows the model to remain as general as possible. In particular, when dealing with two ordinal time series with different state spaces, the introduction of a cross-correlation term could generate non-plausible values. In contrast, in the work of \cite{jentsch2019generalized}, the model is restricted to two binary time series, where both processes take values in $\{0,1\}$. In that setting, cross-correlation is well defined and meaningful.
\end{case}

\begin{case}\label{case3}
In the above definition we assumed a non-parametric pmf for the innovations. In order to achieve parsimony one may assume a parametric model, say e.g. a shifted Binomial distribution, that fully determines the probabilities for each state with much fewer parameters \citep[see,e.g.][]{weiss2019distance}.
\end{case}

With respect to continuous time series, we expect that the methodology remains applicable, provided that a DARMA-class model can be adapted to the continuous setting. The current definition implies that two consecutive observations may take exactly the same value. While this is natural for discrete data, it may seem counterintuitive for continuous models unless interpreted as a form of persistence in the observed values. An alternative approach is to build on the results of \cite{moller2020generalized}, where an appropriate variation operation is introduced to model continuous time series. Following this idea, one could extend the proposed framework to allow at least one of the series to follow a GDARMA model.

\subsection{BDAR($1$) model for ordinal time series}

The main focus of this work is to present a bivariate model for ordinal time series. Thus, at this point we present the proposed methodology up to this special case. 

Let $\boldsymbol{Z}_t$ be
the observed bivariate ordinal processes where $Z_{1t}$ has state space $\mathcal{S}_1=\{s_1,\ldots,s_{d_1}\}$, where $s_1<\ldots<s_{d_1}$ and $Z_{2t}$ has state space $\mathcal{S}_2=\{s_1,\ldots,s_{d_2}\}$, where $s_1<\ldots<s_{d_2}$. The innovation term $\boldsymbol{\epsilon}_t$ is a also a bivariate ordinal process with state spaces $\mathcal{S}_1\times \mathcal{S}_2$. The marginal distributions of individual innovation terms are  $\boldsymbol{p}^{(1)}_{\epsilon}=(p^{(1)}_{\epsilon_{s_1}},\ldots,p^{(1)}_{\epsilon_{s_{d_1}}})^{'}$ and $\boldsymbol{p}^{(2)}_{\epsilon}=(p^{(2)}_{\epsilon_{s_1}},\ldots,p^{(2)}_{\epsilon_{s_{d_2}}})^{'}$ and especially it holds that:  
\begin{eqnarray*}
&\epsilon_{1t} \sim \boldsymbol{p}^{(1)}_{\epsilon},\\ 
&\epsilon_{2t} \sim \boldsymbol{p}^{(2)}_{\epsilon}.
 \end{eqnarray*}
Their joint distribution is defined through an appropriate copula function $C(\cdot,\cdot)$ with dependence parameter $\delta_{\epsilon}$. As random mixtures are concerned, we assume that we have two random variables $\boldsymbol{\alpha}_t=(\alpha_{1t},\alpha_{2t})^{'}$ following marginally a Bernoulli distribution with parameters $\phi_1$ and $\phi_2$ respectively. Their joint distribution is also given by a copula function $C(\cdot,\cdot)$, with dependence parameter $\delta_\alpha$.

\section{Properties}
\label{Sec3}
At this point, we would like to examine the properties of the BDAR($1$) model. We are interested in studying the stationarity conditions, marginal distributions, joint marginal and joint conditional distribution of the two series and the cross-covariance and cross-correlation matrices. We consider the model as defined in Definition \ref{de4}.

\subsection{Stationarity}

\begin{theorem}\label{th1}
 The series is stationary if  $0\leq\phi_1,\phi_2<1$.
\end{theorem}

\begin{proof}
    To examine under what conditions the bivariate model is stationary, we  use the recursion back to $m\to\infty$. Based on the recursion we have that 
\begin{eqnarray*}
\boldsymbol{Z}_t &=&\boldsymbol{\alpha}_{t} \odot \boldsymbol{Z}_{t-1} + \boldsymbol{\beta}_{t} \odot \boldsymbol{\epsilon}_t\\
&=& \boldsymbol{\alpha}_{t}\odot (\boldsymbol{\alpha}_{t-1} \odot\boldsymbol{Z}_{t-2} +\boldsymbol{\beta}_{t-1} \odot \boldsymbol{\epsilon}_{t-1}) +\boldsymbol{\beta}_{t} \odot\boldsymbol{\epsilon}_t
\end{eqnarray*}
and after certain steps we derive that

\begin{equation*}
\boldsymbol{Z}_t =\prod_{d=0}^{m}\boldsymbol{\alpha}_{t-d}\odot \boldsymbol{Z}_{t-(m+1)}+ \sum_{d=1}^{m}\prod_{j=0}^{d-1}\boldsymbol{\alpha}_{t-j} \odot \boldsymbol{\beta}_{t-d} \odot \boldsymbol{\epsilon}_{t-d}+\boldsymbol{\beta}_{t} \odot \boldsymbol{\epsilon}_t
\end{equation*}

We would like $\boldsymbol{Z}_t$ be independent to $\boldsymbol{Z}_{t-(m+1)}$, then $\prod_{d=0}^{m}\boldsymbol{\alpha}_{t-d}$ should goes to $0$, as $m \to \infty$, this means that at least one term of the product should be $0$. To achieve at least one $\boldsymbol{\alpha}_{t-d}$, $d=0,\ldots,m$ to be the zero vector, it should hold that $\pi_{00}=P(a_{1,t-d}=0,a_{2,t-d}=0) \neq 0$. 
It holds that
\begin{equation*}
P(a_{1,t-d}=0,a_{2,t-d}=0;\delta_\alpha) = C(1-\phi_1,1-\phi_2;\delta_\alpha)
\end{equation*}
so, 
$P(a_{1t}=0,a_{2t}=0;\delta_\alpha)=0$ when $1-\phi_1=1-\phi_2=0$ namely $\phi_1=\phi_2=1$. Thus, stationarity is ensured when $0\leq\phi_1<1$ and $0\leq \phi_2<1$. This result is expected as from the properties of univariate DAR($1$) model, the stationary is ensured when $0\leq\phi<1$. 
\end{proof}

\subsection{Marginal distribution}

\begin{theorem}\label{th2}
     Under stationarity assumption it holds that the marginal distribution of $Z_{kt}$ is the same as of $\epsilon_{kt}$, for $k=1,2$. 
\end{theorem}

It suffice to note that marginally for each series we have a simple DAR(1) representation. Based on the properties of DAR(1), it implies that the the marginal distribution of $Z_{kt}$ is the same as of $\epsilon_{kt}$, for $k=1,2$.

It is important to note that the result is not valid in general for the joint distribution as we will show next. 

\subsection{Joint distribution of \texorpdfstring{$\boldsymbol{Z}_t$}{Z_t}}

\begin{theorem}\label{th3}
For the joint distribution
\begin{equation*}
p_{ij} = P(Z_{1t}=i,Z_{2t}=j)
\end{equation*}
it holds that
\begin{equation*}
p_{ij}=\frac{(\pi_{10}+\pi_{01})p^{(1)}_{\epsilon_{i}}p^{(2)}_{\epsilon_{j}}+\pi_{00}p_{\epsilon_{ij}}}{1-\pi_{11}},
\end{equation*}
where 
$p^{(1)}_{\epsilon_{i}}$ and $p^{(2)}_{\epsilon_{j}}$ are the marginal distributions of the innovations and $p_{\epsilon_{ij}}$ the joint distribution of the innovations. Moreover $\pi_{ij}=P(\alpha_{1t}=i,\alpha_{2t}=j)$ i.e. the joint probabilities of the Bernoulli mixing random variables.
\end{theorem}

\begin{proof}
    Denote the joint distribution of the two series as $p_{ij}=P(Z_{1t}=i,Z_{2t}=j)$. Under stationarity, it holds that 
    \begin{equation*}
    P(Z_{1t}=i,Z_{2t}=j)=P(Z_{1,t-1}=i,Z_{2,t-1}=j)=p_{ij}. 
    \end{equation*}
In addition, it holds that: 
\begin{eqnarray*}
P(Z_{1t}=i)&= &{\sum_{j}} p_{ij}=p_{i.} \hspace{2cm} \mbox{and} \\
P(Z_{2t}=j)&=&\sum_{i}p_{ij}=p_{.j}
\end{eqnarray*}
We have four possible ways to observe the pair $(i,j)$ at time $t$. Thus, the joint marginal probability is given by: 
\begin{eqnarray*}
P(Z_{1t}=i,Z_{2t}=j)&=&P(\alpha_{1t}=1,\alpha_{2t}=1)P(Z_{1,t-1}=i,Z_{2,t-1}=j|\alpha_{1t}=1,\alpha_{2t}=1)\\
&+&P(\alpha_{1t}=1,\alpha_{2t}=0)P(Z_{1,t-1}=i,\epsilon_{2t}=j|\alpha_{1t}=1,\alpha_{2t}=0)\\
&+&P(\alpha_{1t}=0,\alpha_{2t}=1)P(\epsilon_{1t}=i,Z_{2,t-1}=j|\alpha_{1t}=0,\alpha_{2t}=1)\\
&+&P(\alpha_{1t}=0,\alpha_{2t}=0)P(\epsilon_{1t}=i,\epsilon_{2t}=j|\alpha_{1t}=0,\alpha_{2t}=0).
\end{eqnarray*}
Under independence of $(\boldsymbol{Z}_s)_{s<t}$ and  $\boldsymbol{\epsilon}_t$, stationarity and assuming that marginal distribution of $Z_{kt}$ is the same as of $\epsilon_{kt}$, for $k=1,2$ we have that
\begin{equation*}
p_{ij}=\pi_{11}p_{ij}+\pi_{10}p^{(1)}_{\epsilon_{i}}p^{(2)}_{\epsilon_{j}}+\pi_{01}p^{(1)}_{\epsilon_{i}}p^{(2)}_{\epsilon_{j}}+\pi_{00}p_{\epsilon_{ij}}
\end{equation*}
which can be also written as
\begin{equation*}
p_{ij}(1-\pi_{11})=(\pi_{10}+\pi_{01})p^{(1)}_{\epsilon_{i}}p^{(2)}_{\epsilon_{j}}+\pi_{00}p_{\epsilon_{ij}}
\end{equation*}
that gives the required relationship
\begin{equation}
p_{ij}=\frac{(\pi_{10}+\pi_{01})p^{(1)}_{\epsilon_{i}}p^{(2)}_{\epsilon_{j}}+\pi_{00}p_{\epsilon_{ij}}}{1-\pi_{11}}.
\label{formula}
\end{equation}
So, even though $Z_{kt}$ and $\epsilon_{kt}$ for $k=1,2$ have the same marginal distribution, the marginal distribution of $\boldsymbol{Z}_t$ differs from the marginal distribution of $\boldsymbol{\epsilon}_t$. The only case for which $\boldsymbol{Z}_t$ and $\boldsymbol{\epsilon}_t$ share the same marginal distribution, i.e. it holds that 
$p_{ij} = p_{\epsilon_{ij}}$,
is when there is one common random mechanism that describes both of the series. Namely, we have the representation
\begin{align*}
\alpha_{t} &\sim Bernoulli(\phi)\\
\begin{bmatrix} Z_{1t} \\ Z_{2t} \end{bmatrix}&=\alpha_{t} \odot \begin{bmatrix} Z_{1,t-1} \\ Z_{2,t-1} \end{bmatrix}+
 (1-\alpha_{t}) \odot \begin{bmatrix} \epsilon_{1t} \\ \epsilon_{2t} \end{bmatrix} \nonumber \\
\end{align*}
Then, it holds that:
\begin{align*}
p_{ij}&=\phi P(Z_{1,t-1}=i,Z_{2,t-1}=j)+(1-\phi)P(\epsilon_{1t}=i,\epsilon_{2t}=j)\\
&=\phi p_{ij}+(1-\phi)p_{\epsilon_{ij}}
\end{align*}
which gives that $p_{ij}=p_{\epsilon_{ij}}$. 
One can see that in such case $\pi_{01}=\pi_{10}=0$ in \ref{formula}.
\end{proof}

\subsection{Joint Conditional probabilities}
In this part we would like to define the joint distribution of $\boldsymbol{Z}_t$ conditional to the vector $\boldsymbol{Z}_{t-1}$, as we have considered a BDAR($1$) model. These probabilities are useful for the estimation of the model as it will be discussed later. To define the joint conditional probabilities, we also have to take into consideration the four possible outcomes from the joint distribution of the two Bernoulli random variables. 

\begin{theorem}\label{th4}
The joint conditional to the past probabilities are given by: 
\begin{align}
&P(Z_{1t}=i,Z_{2t}=j|Z_{1,t-1}=s,Z_{2,t-1}=\ell) \nonumber\\
&=\pi_{11}I((i,j)=(s,\ell))+\pi_{00}P(\epsilon_{1t}=i,\epsilon_{2t}=j)+\pi_{10}P(Z_{1,t-1}=i,\epsilon_{2t}=j) \nonumber\\
&+\pi_{01}P(\epsilon_{1t}=i,Z_{2,t-1}=j) \nonumber \\
&=\pi_{11}I((i,j)=(s,\ell))+\pi_{00}p_{\epsilon_{ij}}+\pi_{10}I(i=s)p^{(2)}_{\epsilon_{j}}+\pi_{01}p^{(1)}_{\epsilon_i}I(j=\ell) \nonumber \\
&=(1-\pi_{10}-\pi_{01}-\pi_{00})I((i,j)=(s,\ell))+\pi_{00}p_{e_{ij}}+\pi_{10}I(i=s)p^{(2)}_{\epsilon_{j}}+\pi_{01}p^{(1)}_{\epsilon_i}I(j=\ell)
\label{Eq:Joint conditional probabilities}
\end{align}
where $I(\cdot)$ indicator function that takes value $1$ when $\boldsymbol{Z}_t=\boldsymbol{Z}_{t-1}$ and $0$ otherwise. 
\end{theorem}

The proof is based on enumerating all possible outcomes to move from the current values $s$ and $\ell$ to the new ones $i $ and $j$. For example if $i \ne s$ and $j \ne \ell$ then this can happen only if we  have taken new values from the innovation distributions. On the other hand for the case when 
$i= s$ and $j = \ell$, then we can have 4 cases, staying at the same values, generating new values from the innovation and also the two cases when only one of the two values is the same and the other is taken from the innovations. Summing all possible cases gives the above formula.

\subsection{Cross-Covariance \& Cross-Correlation matrices}

Let $\boldsymbol{\Gamma}(k)$ denotes the cross-covariance matrix of $Z_{1t}$ and $Z_{2t}$ at lag $k$, with elements $\gamma_{rs}(k)=Cov(Z_{rt},Z_{st})=E((Z_{rt}-\mu_r)(Z_{s,t-k}-\mu_s)),~ s,r=1,2$:
 
\begin{align*}
\boldsymbol{\Gamma}{(k)}=\begin{pmatrix}
Cov(Z_{1t},Z_{1,t-k}) & Cov(Z_{1t},Z_{2,t-k})\\
Cov(Z_{1,t-k},Z_{2t}) & Cov(Z_{2t},Z_{2,t-k})
\end{pmatrix}= \begin{pmatrix}
\gamma_{11}(k) & \gamma_{12}(k)\\   
\gamma_{21}(k) & \gamma_{22}(k)  
\end{pmatrix}
\end{align*}
We also define the cross-correlation matrix of $Z_{1t}$ and $Z_{2t}$ at lag $k$, denoted by $\rho(k)$. Each element of $\boldsymbol{\rho}(k)$ is: 
$$\rho_{rs}(k)=\frac{E((Z_{rt}-\mu_r)(Z_{s,t-k}-\mu_s))}{\sqrt{E((Z_{rt}-\mu_r)^2)}\sqrt{E((Z_{s,t-k}-\mu_s)^2)}}=\frac{\gamma_{rs}(k)}{\sqrt{\gamma_{rr}(0)}\sqrt{\gamma_{ss}(0)}},~ s,r=1,2:$$

\begin{align*}
\boldsymbol{\rho}{(k)}=\begin{pmatrix}
Cor(Z_{1t},Z_{1,t-k}) & Cor(Z_{1t},Z_{2,t-k})\\
Cor(Z_{1,t-k},Z_{2t}) & Cor(Z_{2t},Z_{2,t-k})
\end{pmatrix}= \begin{pmatrix}
\rho_{11}(k) & \rho_{12}(k)\\   
\rho_{21}(k) & \rho_{22}(k)  
\end{pmatrix}
\end{align*}
When $r=s$, then $\boldsymbol{\Gamma}(k)$ is the auto-covariance matrix and $\boldsymbol{\rho}(k)$ the auto-correlation matrix of the respective time series.\\ 
For the following calculations it holds that, 
\begin{align*}
E(a_{1t}a_{2t})&=\phi_{12}\\
E(\beta_{1t}\beta_{2t})&=Cov(\beta_{1t},\beta_{2t})+E(\beta_{1t})E(\beta_{2t})\\
&=Cov(a_{1t},a_{2t})+E(1-a_{1t})E(1-a_{2t})\\
&=\phi_{12}+(1-\phi_1)(1-\phi_2)=(1-\phi_1-\phi_2+\phi_{12})\\
E(\alpha_{1t}^2)&=V(\alpha_{1t})+E(\alpha_{1t})^{2}=\phi_1(1-\phi_1)+\phi_1^{2}=\phi_1\\
E(\alpha_{2t}^2)&=V(\alpha_{2t})+E(\alpha_{2t})^{2}=\phi_2(1-\phi_2)+\phi_2^{2}=\phi_2\\
E(\alpha_{1t}\beta_{2t})&=E(\alpha_{1t}(1-\alpha_{2t}))=E(\alpha_{1t})-E(\alpha_{1t}\alpha_{2t})=\phi_1-\phi_{12}\\
E(\alpha_{2t}\beta_{1t})&=E(\alpha_{2t}(1-\alpha_{1t}))=E(\alpha_{2t})-E(\alpha_{1t}\alpha_{2t})=\phi_2-\phi_{12}\
\end{align*}

\begin{itemize}
    \item For lag $0$ we have that
\begin{align*}
\gamma_{11}{(0)}&=E((Z_{1t}-\mu_1)(Z_{1t}-\mu_1))=
E((Z_{1t}-\mu_1)^2)\\&=V(Z_{1t})=V(\epsilon_{1t}) \hspace{1cm} \mbox{and} \\
\rho_{11}(0)&=1
\end{align*}

Similarly, 
\begin{align*}
\gamma_{22}{(0)}&=E((Z_{2t}-\mu_2)(Z_{2t}-\mu_2))=
E((Z_{2t}-\mu_2)^2)\\&=V(Z_{2t})=V(\epsilon_{2t}) \hspace{1cm} \mbox{and}\\
\rho_{22}(0)&=1
\end{align*}
We also get that
\begin{align*}
\gamma_{12}(0)&=E(Z_{1t}Z_{2t})-E(Z_{1t})E(Z_{2t})\\
E(Z_{1t}Z_{2t})&=E[(\alpha_{1t}Z_{1,t-1}+\beta_{1t}\epsilon_{1t})(\alpha_{2t}Z_{2,t-1}+\beta_{2t}\epsilon_{2t})]\\
&=E(\alpha_{1t}\alpha_{2t}Z_{1,t-1}Z_{2,t-1})+E(\alpha_{1t}Z_{1,t-1}\beta_{2t}\epsilon_{2t})+E(\beta_{1t}\epsilon_{1t}\alpha_{2t}Z_{2,t-1})\\
&+E(\beta_{1t}\beta_{2t}\epsilon_{1,t-1}\epsilon_{2,t-1})\\
&=E(\alpha_{1t}\alpha_{2t})E(Z_{1,t-1}Z_{2,t-1})+E(\alpha_{1t}\beta_{2t})E(Z_{1,t-1})E(\epsilon_{2t})\\
&+E(\alpha_{2t}\beta_{1t})E(Z_{2,t-1})E(\epsilon_{1t})+E(\beta_{1t}\beta_{2t})E(\epsilon_{1,t-1}\epsilon_{2,t-1})\\
&=\phi_{12}E(Z_{1t}Z_{2t})+(\phi_1-\phi_{12})\mu_1\mu_2+(\phi_2-\phi_{12})\mu_2\mu_1\\
&+(1-\phi_1-\phi_2+\phi_{12})E(\epsilon_{1,t-1}\epsilon_{2,t-1})
\end{align*}
and rearranging the terms we get
\begin{align*}
E(Z_{1t}Z_{2t})&=\frac{\mu_1\mu_2(\phi_1-\phi_{12}+\phi_2-\phi_{12})+(1-\phi_1-\phi_2+\phi_{12})E(\epsilon_{1,t-1}\epsilon_{2,t-1})}{1-\phi_{12}}
\end{align*}
and then we get that
\hspace{-2cm}
\begin{align*}
 \gamma_{12}(0)&=\frac{\mu_1\mu_2(\phi_1-\phi_{12}+\phi_2-\phi_{12})+(1-\phi_1-\phi_2+\phi_{12})E(\epsilon_{1,t-1}\epsilon_{2,t-1})}{1-\phi_{12}}-\mu_1\mu_2\\ 
&=\frac{(1-\phi_1-\phi_2+\phi_{12})(E(\epsilon_{1,t-1}\epsilon_{2,t-1})-\mu_1\mu_2)}{1-\phi_{12}} \hspace{1cm} \mbox{and}\\
\rho_{12}(0)&=\rho_{21}(0)=\frac{\gamma_{12}(0)}{\sqrt{\gamma_{11}(0)}\sqrt{\gamma_{22}(0)}}
\end{align*}

\item While for lag $k$ we have that 
\begin{align*}
\gamma_{11}(k)&=E(Z_{1t}Z_{1,t-k})- E(Z_{1t})E(Z_{1,t-k})\\
E(Z_{1t}Z_{1,t-k})&=E(Z_{1,t-k}(\alpha_{1t}Z_{1,t-1}+\beta_{1t}\epsilon_{1t}))\\
&=E(\alpha_{1t}Z_{1,t-k}Z_{1,t-1})+E(\beta_{1t}Z_{1,t-k}\epsilon_{1t})\\
&=E(\alpha_{1t})E(Z_{1,t-k}Z_{1,t-1})+E(\beta_{1t})E(Z_{1,t-k})E(\epsilon_{1t})\\
&=\phi_1E(Z_{1,t-k}Z_{1,t-1})+(1-\phi_1)\mu_1^{2}
\end{align*}
and then we get that 
\begin{align*}
\gamma_{11}(k)&=\phi_1E(Z_{1,t-k}Z_{1,t-1})+(1-\phi_1)\mu_1^{2} -\mu_1^{2}\\
&=\phi_1(E(Z_{1,t-k}Z_{1,t-1})-\mu_1^{2})\\
&=\phi_1\gamma_{11}(k-1)
\end{align*}

$$\rho_{11}(k)=\frac{\gamma_{11}(k)}{\gamma_{11}(0)}$$

Similarly, 
$$\gamma_{22}(k)=\phi_2\gamma_{22}(k-1)$$
$$\rho_{22}(k)=\frac{\gamma_{22}(k)}{\gamma_{22}(0)}$$

For the cross-correlation at lag $k$, we have that
\begin{align*}
\gamma_{12}(k)&=E(Z_{1t}Z_{2,t-k})-E(Z_{1t})E(Z_{2,t-k})\\
E(Z_{1t}Z_{2,t-k})&=E[Z_{2,t-k}(\alpha_{1t}Z_{1,t-1}+\beta_{1t}\epsilon_{1t})]\\
&=E(\alpha_{1t}Z_{2,t-k}Z_{1,t-1})+E(\beta_{1t}Z_{2,t-k}\epsilon_{1t})\\
&=E(\alpha_{1t})E(Z_{2,t-k}Z_{1,t-1})+E(\beta_{1t})E(Z_{2,t-k})E(\epsilon_{1t})\\
&=\phi_1E(Z_{2,t-k}Z_{1,t-1}) +(1-\phi_1)\mu_2\mu_1
\end{align*}

and then we get
\begin{align*}
\gamma_{12}(k)&=\phi_1E(Z_{2,t-k}Z_{1,t-1}) +(1-\phi_1)\mu_2\mu_1-\mu_1\mu_2\\
&=\phi_1(E(Z_{2,t-k}Z_{1,t-1})-\mu_1\mu_2)\\
&=\phi_1\gamma_{12}(k-1)
\end{align*}
$$\rho_{12}(k)=\frac{\gamma_{12}(k)}{\sqrt{\gamma_{11}(0)}\sqrt{\gamma_{22}(0)}}$$

Similarly, 
$$\gamma_{21}(k)=\phi_2 \gamma_{21}(k-1)$$  
$$\rho_{21}(k)=\frac{\gamma_{21}(k)}{\sqrt{\gamma_{11}(0)}\sqrt{\gamma_{22}(0)}}$$
\end{itemize}

\section{Estimation}
\label{Sec4}

Estimation of DAR($1$) model is based on the maximization of the conditional log-likelihood
\begin{align*}
\ell(\boldsymbol{\theta})=\sum_{t=2}^{T}\log(P(Z_t=z_t|Z_{t-1}=z_{t-1}))
\end{align*}
More specifically, for the case of categorical time series,  under DAR($1$) model, the conditional probabilities are given by :
\begin{align*}
   P(Z_t=i|Z_{t-1}=j)=(1-\phi)P(\epsilon_t=i) +\phi I(i=j)
\end{align*}
where $I(\cdot)$ an indicator function that takes value $1$ when $Z_t=Z_{t-1}$ and $0$ otherwise. In addition, it is valid that $P(\epsilon_t=i)=P(Z_t=i)=p_{\epsilon_i}$. Then the vector of parameters to be estimated is  $\boldsymbol{\theta}=(\phi,\boldsymbol{p}_{\epsilon})^{'}$. The parameter $\phi$ should takes values in $[0,1)$, for assuring stationarity. For the vector of probabilities of success, it is valid that $p_{\epsilon_i}\in (0,1]$, for $i=1,\ldots,d$ and  that $\sum_{i=1}^{d}p_{\epsilon_i}=1$, where $d$ is the number of possible states for the observed categorical process, so maximization should be under this restriction.

For the estimation of BDAR($1$) model we follow the same approach by maximizing the conditional log-likelihood. However, under the assumption of a bivariate model we need the joint conditional probabilities. As it has been described in \ref{Eq:Joint conditional probabilities}, $P(\boldsymbol{z}_t|\boldsymbol{z}_{t-1})=P(z_{1t},z_{2t}|z_{1,t-1},z_{2,t-1})$ and then the conditional log-likelihood is given by: 
\begin{align*}
\ell(\boldsymbol{\theta})=\sum_{t=2}^{T}\log(P(\boldsymbol{Z}_t=\boldsymbol{z}_t|\boldsymbol{Z}_{t-1}=\boldsymbol{z}_{t-1}))
\end{align*}

For the case of two ordinal time series, which is our interest, the vector $\boldsymbol{\theta}$ includes the following parameters $\boldsymbol{\theta}=(\phi_1,\phi_2,\delta_{\alpha},\delta_{\epsilon}, \boldsymbol{p}^{(1)}_{\epsilon},\boldsymbol{p}^{(2)}_{\epsilon})^{'}$, Similar to  DAR($1$) example, for $\phi_1$ and $\phi_2$ it holds that $0\leq \phi_1<1, \quad 0\leq \phi_2<1$, to ensure stationarity. For the parameters of multinomial distributions, each probability of success $p^{(k)}_{\epsilon_{s_i}}$ should be in range $[0,1]$, while we have the restrictions $\sum_{i=1}^{d_1}p^{(1)}_{\epsilon_{s_i}}=1$ and $\sum_{i=1}^{d_2}p^{(2)}_{\epsilon_{s_i}}=1$. For copulas parameters $\delta_{\alpha}$ and $\delta_{\epsilon}$ any restriction depends on the chosen copula functions.

\section{Simulations}
\label{Sec5}
In this part we provide a simulation study to examine the performance of the proposed estimators for the BDAR model. More particularly, we assume a  BDAR($1$) for two ordinal ordinal time series $Z_{1t}$ and $Z_{2t}$ with three possible states for each of them $\mathcal{S}_1=\mathcal{S}_2=(s_1,s_2,s_3)$, where $s_1<s_2<s_3$. The innovation terms $\epsilon_{1t}$ and $\epsilon_{2t}$ are marginally assumed to be distributed according to a multinomial distribution, while their joint distribution is defined through a Gumbel copula with parameter $\delta_{\epsilon}=2$. We set:

\begin{align*}
&\epsilon_{1t} \sim \boldsymbol{p}^{(1)}_{\epsilon}=( p^{(1)}_{\epsilon_{1}}=0.15,p^{(1)}_{\epsilon_{2}}=0.6,p^{(1)}_{\epsilon_{3}}=0.25)\\   
&\epsilon_{2t} \sim\boldsymbol{p}^{(2)}_{\epsilon}=( p^{(2)}_{\epsilon_{1}}=0.2,p^{(2)}_{\epsilon_{2}}=0.3,p^{(2)}_{\epsilon_{3}}=0.5)~~~~~\text{and}\\
P(\epsilon_{1t},\epsilon_{2t})&=C_G(F_{\epsilon_{1t}}(\epsilon_{1t}),F_{\epsilon_{2t}}(\epsilon_{2t});\delta_{\epsilon})-C_G(F_{\epsilon_{1t}}(\epsilon_{1t}-1),F_{\epsilon_{2t}}(\epsilon_{2t});\delta_{\epsilon}) \\
    &-C_G(F_{\epsilon_{1t}}(\epsilon_{1t}),F_{\epsilon_{2t}}(\epsilon_{2t}-1);\delta_{\epsilon})+
    C_G(F_{\epsilon_{1t}}(\epsilon_{1t}-1),F_{\epsilon_{2t}}(\epsilon_{2t}-1);\delta_{\epsilon}),
 \end{align*}

where, 
\begin{align}
\label{Eq: Gumbel}
C_{G}(u,v;\delta)=\exp \left[ -((-\log(u))^{\delta}+(-\log(v))^{\delta})^{\frac{1}{\delta}}\right],  \quad \delta \in [1,\infty)
\end{align}
for $u=F_1(\epsilon_{1t})$ and $v=F_2(\epsilon_{2t})$ and $F_k(.),  k=1,2$ the corresponding cumulative distribution function.  

For the mixture mechanisms, the two Bernoulli random variables are also assumed to jointly follow a Gumbel copula with dependence parameter $\delta_{\alpha}=2$, so 
\begin{align*}
&a_{1t} \sim Bernoulli(\phi_1=0.4),\\ 
&a_{2t} \sim Bernoulli(\phi_2=0.25), \hspace{1cm} \mbox{and}\\
P(\alpha_{1t},\alpha_{2t})&=C_G(F_{\alpha_{1t}}(\alpha_{1t}),F_{\alpha_{2t}}(\alpha_{2t});\delta_{\alpha})-C_G(F_{\alpha_{1t}}(\alpha_{1t}-1),F_{\alpha_{2t}}(\alpha_{2t});\delta_{\alpha})\\
    &-C_G(F_{\alpha_{1t}}(\alpha_{1t}),F_{\alpha_{2t}}(\alpha_{2t}-1);\delta_{\alpha})+
    C_G(F_{\alpha_{1t}}(\alpha_{1t}-1),F_{\alpha_{2t}}(\alpha_{2t}-1);\delta_{\alpha}).
\end{align*}

We are interested in examining the performance of the estimators under different sample sizes. More specifically, we assume three scenarios: time series of length $T=100, 500$ and $1000$. In each case we simulate $500$ replicates. The results are presented in Figures \ref{Figure:BDAR_Sim1_Part1} and \ref{Figure:BDAR_Sim1_Part2}. Especially for Figure \ref{Figure:BDAR_Sim1_Part2}, the results are presented in logarithmic scale. Based on them, as we expect, it seems that variability decreases as sample size increases. As bias is concerned, it seems that for all sample sizes the medians of the boxplots of all parameters are very close to the true values. Nevertheless, for $T=100$, it seems that the variance of copulas' parameters is essentially higher compared to other sample sizes.

\begin{figure}[t]
\begin{center}
\includegraphics[scale=0.45]{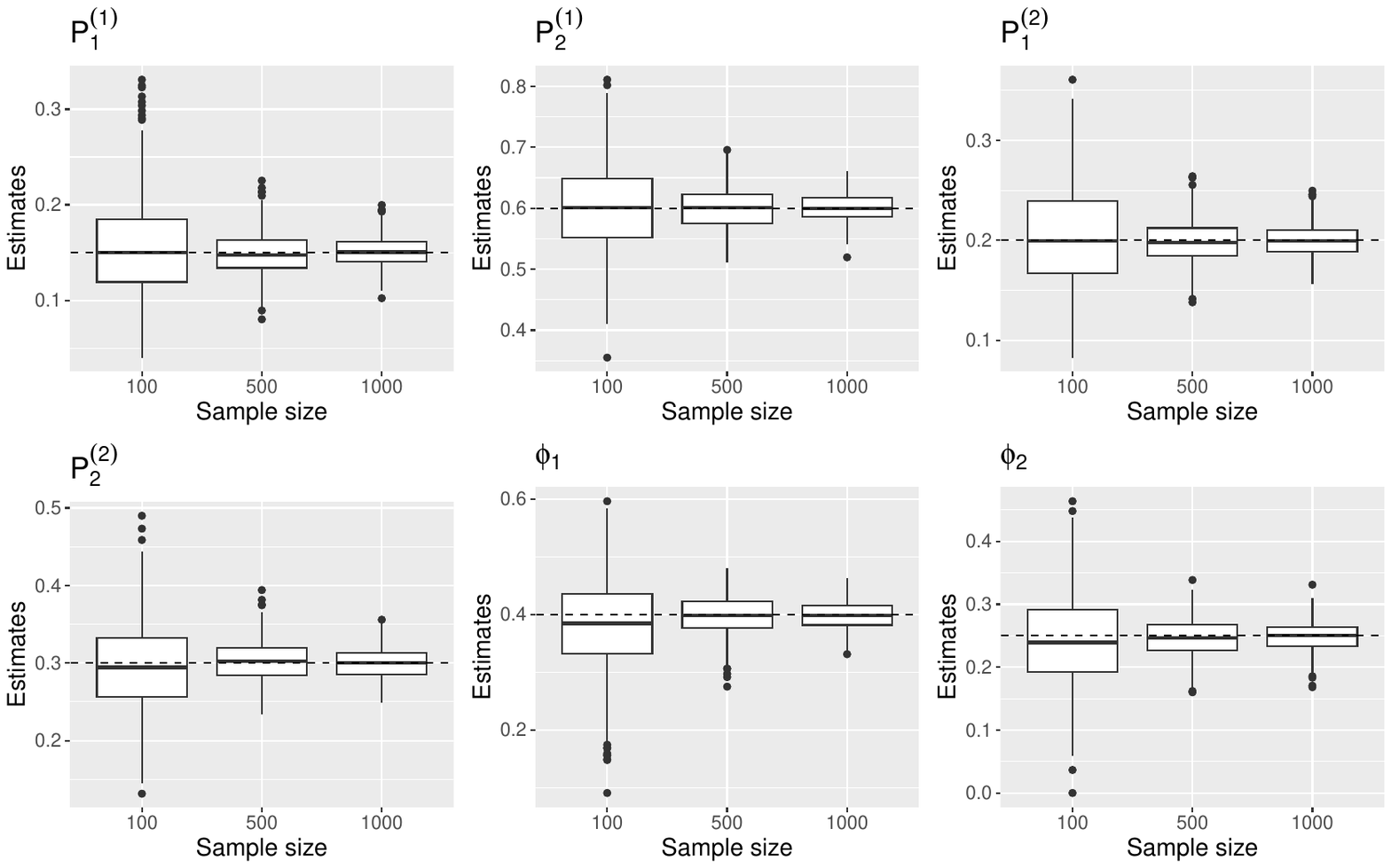}
\caption{\label{Figure:BDAR_Sim1_Part1}Boxplots of estimates of parameters of multinomial and Bernoulli for different sample sizes. Dashed line is the true value of the parameter}
\end{center}
\end{figure}

\begin{figure}[t]
\begin{center}
\includegraphics[scale=0.40]{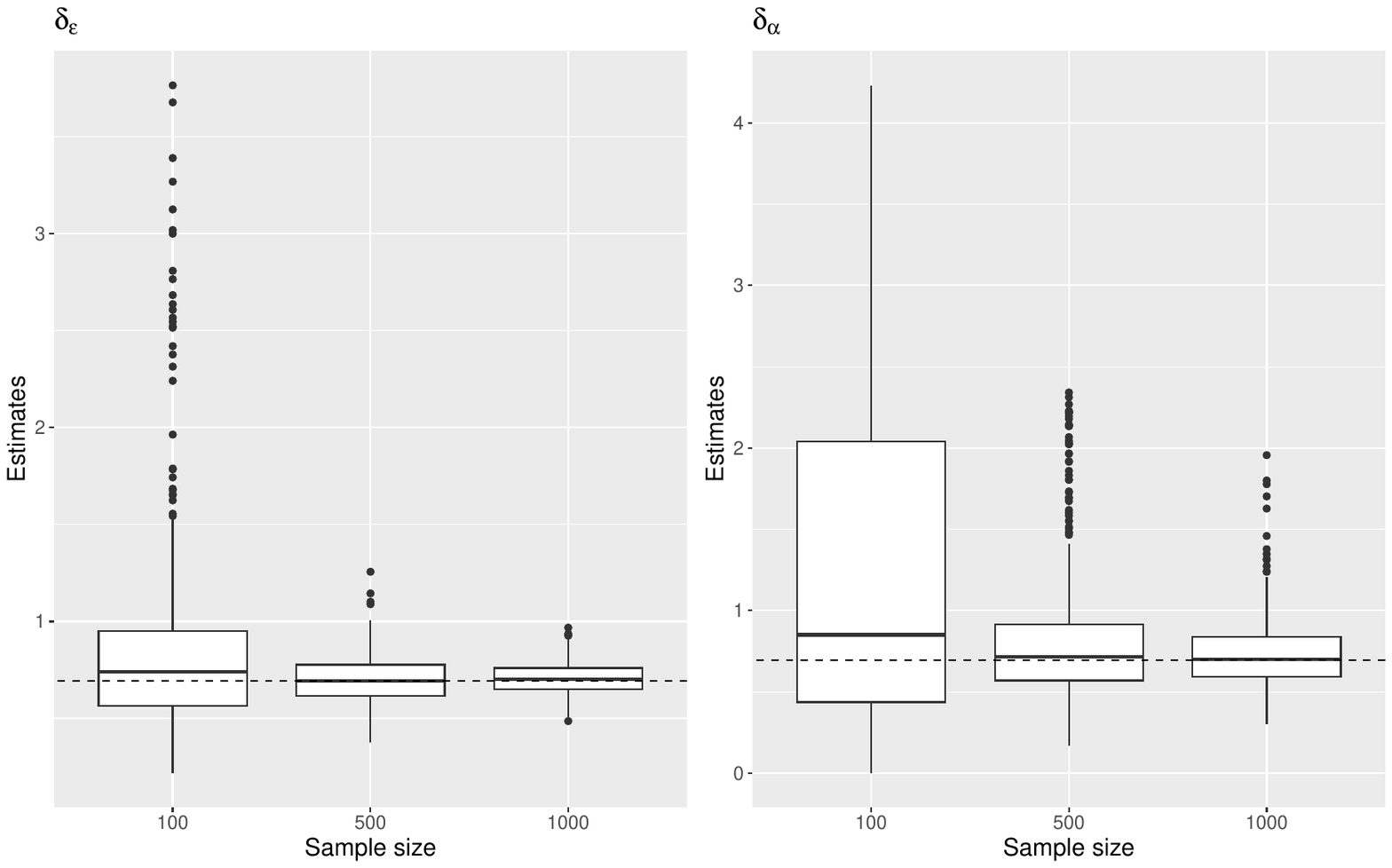}
\caption{\label{Figure:BDAR_Sim1_Part2}Boxplots of estimates of copulas' parameters for different sample sizes in logarithmic scale. Dashed line is the true value of the parameter.}
\end{center}
\end{figure}

We also provide the Mean Absolute Error (MAE) for the vector of parameters under the three scenarios of different sample sizes. The results are presented in Figure \ref{Figure:BDAR_Sim1_MAE} in logarithmic scale. According to it, we conclude that as sample size increases the MAE decreases, which is the desirable property. 

\begin{figure}[t]
\begin{center}
\includegraphics[scale=0.35]{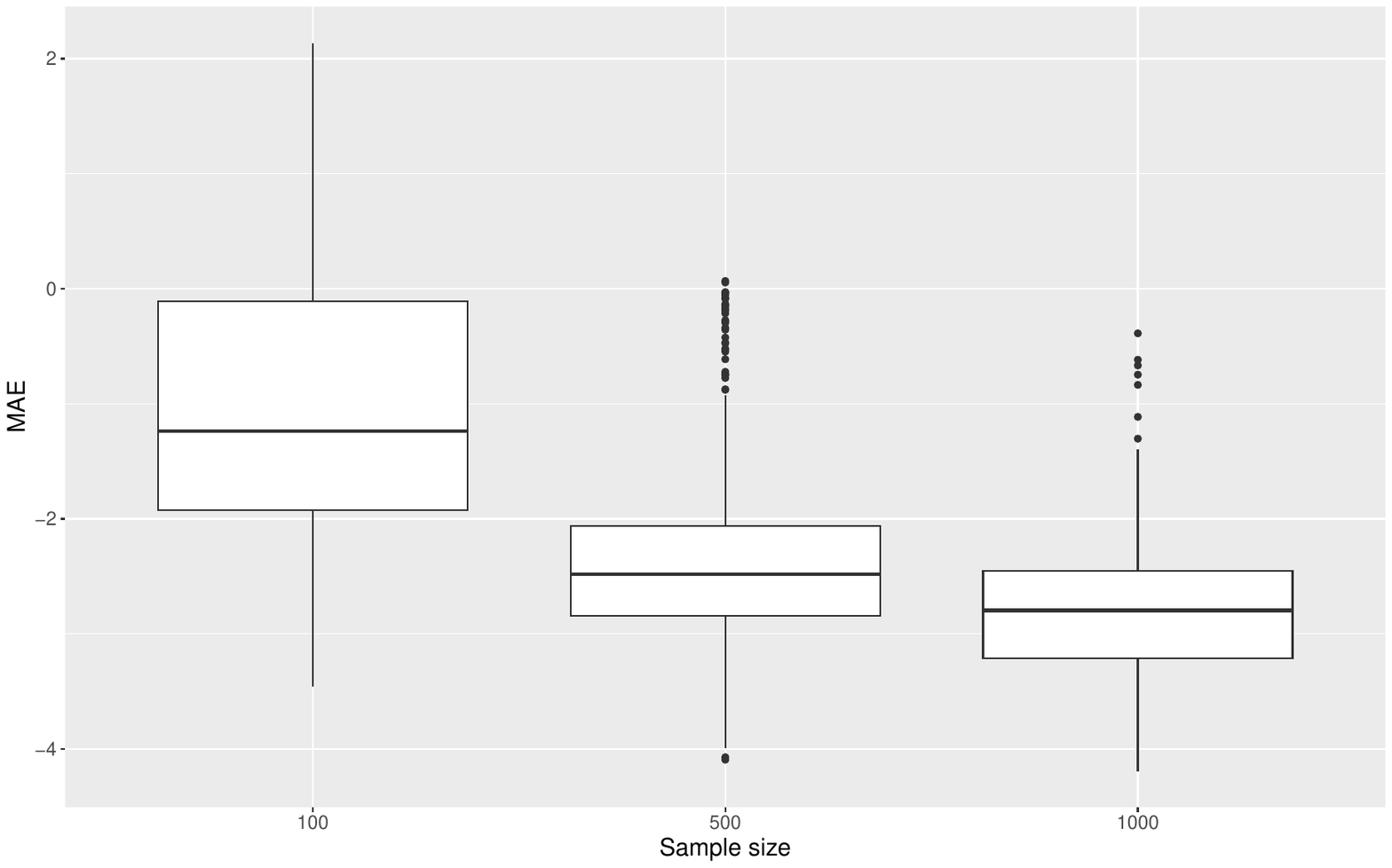}
\caption{\label{Figure:BDAR_Sim1_MAE}Boxplots of MAE for different sample sizes, in logarithmic scale.}
\end{center}
\end{figure}

Taking everything into consideration, we can conclude that as sample size increases the model has the desired properties of consistent and unbiased estimators. Note that for small sample sizes ($T=100$) the variance of copulas' parameter can be large.

\section{Application: Unemployment state Slovakia-Czech Republic }
\label{Sec6}

\subsection{About the data} Unemployment is a social and economic phenomenon that impacts not only governments and societies but also families and individuals. Studying the unemployment rate is therefore of particular importance, as it can reveal the factors that exacerbate the problem and inform policy decisions aimed at alleviating it. Since unemployment is inherently a social phenomenon, it is reasonable to assume that the unemployment rates of different countries may be associated, with one influencing the other. Motivated by this assumption, we focus on jointly modeling the unemployment rates of two neighboring countries: Slovakia and the Czech Republic. The dataset covers quarterly observations from 1998 to 2023. However, aiming at making the values more comparable, we have discretized them based on the quantiles of the whole set of values of the examined countries. Thus, we transform the continuous time series into ordinal time series based on the following rule:

\begin{equation} 
 Z_{kt}= \left\{
\begin{array}{ll}
      1, & 1.9\leq Y_{kt} \leq 5.9\\
      2, & 5.9 < Y_{kt} \leq 7.7\\
      3, & 7.7 < Y_{kt} \leq 12.75 \\
      4, & 12.75 < Y_{kt} \leq 19.9 \\
\end{array} 
\right. 
\label{Eq.Discrete var}
\end{equation}
where $Y_{kt}$ denotes the original continuous time series for each country $k=1,2$ (Slovakia and Czech Republic respectively) and $Z_{kt}$ denotes the ordinal time series for each country, $k=1,2$, $t=1,\ldots,104$. In Figure \ref{Fig:Application} the ordinal time series are presented. 

\begin{figure}[t]
\begin{center}
\includegraphics[scale=0.40]{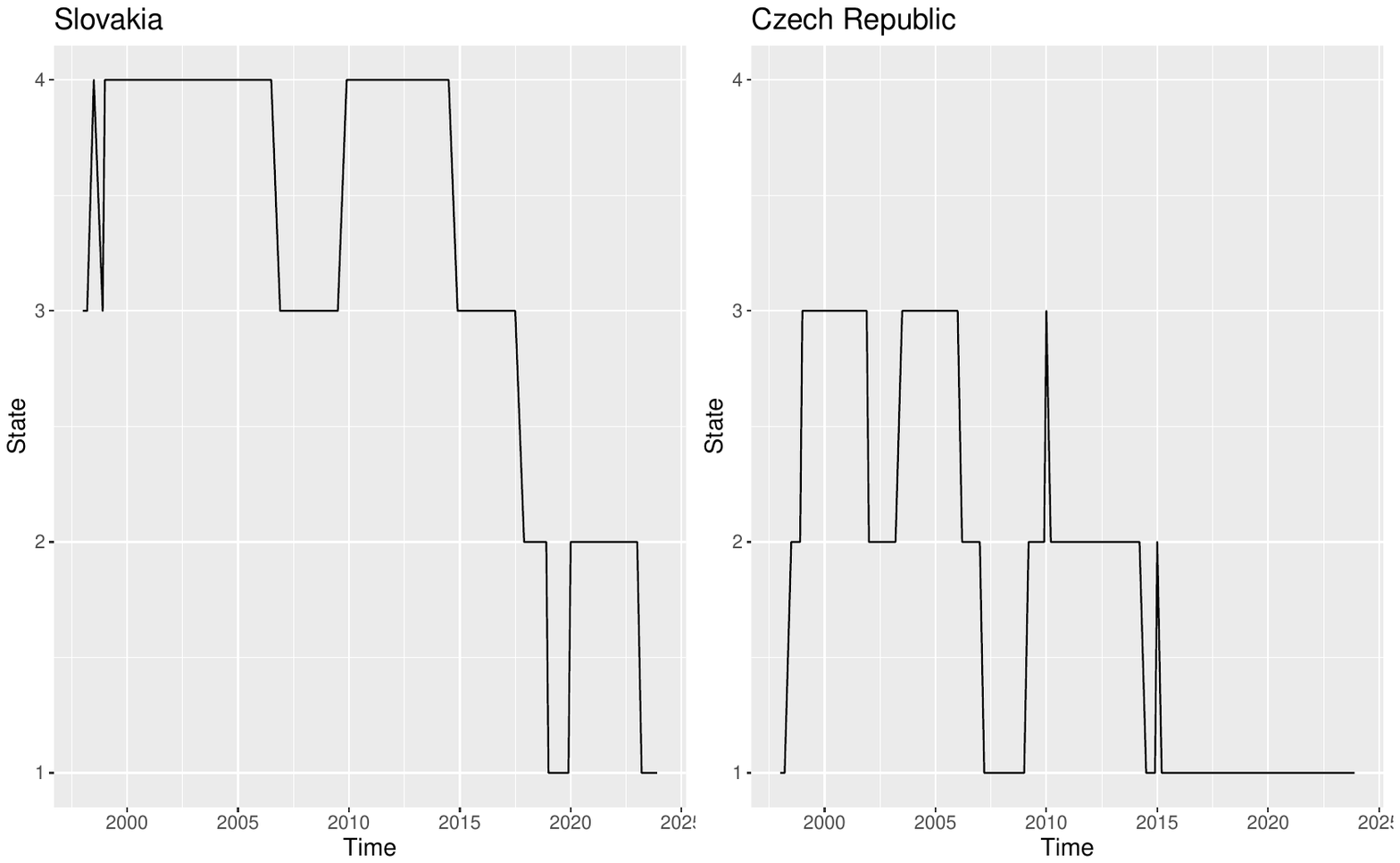}\caption{\label{Fig:Application}Unemployment state per quarter for Slovakia and Czech Republic from 1998 to 2023.}
\end{center}
\end{figure}

\subsection{Model selection \& Estimation}

Based on Kendall's $\tau$, the two time series have correlation $\tau=0.75$ indicating an interesting association. Based on this result we assume that a bivariate model will be more adequate than modelling each series individually. To jointly model the two ordinal time series we use the BDAR($1$) model, which is a plausible choice as the two time series present high autocorrelation at lag $1$, also based on  Kendall's tau, $\tau_1=0.91$ for Slovakia and $\tau_2=0.88$ for Czech Republic. For the joint distribution of random mixtures and the joint distribution of the innovations we will assume Frank copula:
\begin{align}
C_F(u,v;\delta)=-\frac{1}{\delta}\log \left[ 1+ \frac{(\exp(-\delta u)-1)(\exp(-\delta v)-1)}{\exp(-\delta)-1} \right],\quad \delta \in \mathbb{R} \setminus \{0\}  
\label{Equation: Frank copula}
\end{align}
to describe the joint distributions. However, we can also use other copulas to define the joint distributions, while we can also consider different copulas for innovations and random mechanisms. In addition, we assume different special cases of BDAR($1$) with the purpose of finding the most appropriate model. More specifically, we consider the following models:

\begin{itemize}
    \item Model 1: Independent time series 
\begin{align*}
\begin{bmatrix} Z_{1t} \\ Z_{2t}\end{bmatrix}&=\begin{bmatrix} \alpha_{1t} \\ \alpha_{2t}\end{bmatrix} \odot \begin{bmatrix} Z_{1,t-1} \\ Z_{2,t-1}\end{bmatrix}+\begin{bmatrix} 1-\alpha_{1t} \\ 1-\alpha_{2t}\end{bmatrix} \odot \begin{bmatrix} \epsilon_{1t} \\ \epsilon_{2t}\end{bmatrix}\\
\epsilon_{1t} &\sim \boldsymbol{p}^{(1)}_{\epsilon}=( p^{(1)}_{\epsilon_{1}},p^{(1)}_{\epsilon_{2}},p^{(1)}_{\epsilon_{3}},p^{(1)}_{\epsilon_{4}})\\ 
\epsilon_{2t} &\sim\boldsymbol{p}^{(2)}_{\epsilon}=( p^{(2)}_{\epsilon_{1}},p^{(2)}_{\epsilon_{2}},p^{(2)}_{\epsilon_{3}})\\
\alpha_{1t}&\sim Bernoulli(\phi_1)\\
\alpha_{2t}&\sim Bernoulli(\phi_2)
\end{align*}

    \item Model 2: Time series with common random mechanism and dependent innovations
  \begin{align*}  
  \begin{bmatrix} Z_{1t} \\ Z_{2t}\end{bmatrix}&= \begin{bmatrix}\alpha_t \\ \alpha_t \end{bmatrix} \odot \begin{bmatrix} Z_{1,t-1} \\ Z_{2,t-1}\end{bmatrix}+ \begin{bmatrix} 1-\alpha_t\\ 1-\alpha_t \end{bmatrix} \odot \begin{bmatrix} \epsilon_{1t} \\ \epsilon_{2t}\end{bmatrix}\\
\epsilon_{1t} &\sim \boldsymbol{p}^{(1)}_{\epsilon}=( p^{(1)}_{\epsilon_{1}},p^{(1)}_{\epsilon_{2}},p^{(1)}_{\epsilon_{3}},p^{(1)}_{\epsilon_{4}})\\ 
\epsilon_{2t} &\sim\boldsymbol{p}^{(2)}_{\epsilon}=( p^{(2)}_{\epsilon_{1}},p^{(2)}_{\epsilon_{2}},p^{(2)}_{\epsilon_{3}})\\
P(\epsilon_{1t},\epsilon_{2t})&=C_F(F_{\epsilon_{1t}}(\epsilon_{1t}),F_{\epsilon_{2t}}(\epsilon_{2t});\delta_{\epsilon})-C_F(F_{\epsilon_{1t}}(\epsilon_{1t}-1),F_{\epsilon_{2t}}(\epsilon_{2t});\delta_{\epsilon})\\
    &-C_F(F_{\epsilon_{1t}}(\epsilon_{1t}),F_{\epsilon_{2t}}(\epsilon_{2t}-1);\delta_{\epsilon})+
    C_F(F_{\epsilon_{1t}}(\epsilon_{1t}-1),F_{\epsilon_{2t}}(\epsilon_{2t}-1);\delta_{\epsilon})\\
    \alpha_t&\sim Bernoulli(\phi)
\end{align*}

    \item Model 3: Time series with different independent random mechanisms, but dependent innovations
    \begin{align*}
\begin{bmatrix} Z_{1t} \\ Z_{2t}\end{bmatrix}&=\begin{bmatrix} \alpha_{1t} \\ \alpha_{2t}\end{bmatrix} \odot \begin{bmatrix} Z_{1,t-1} \\ Z_{2,t-1}\end{bmatrix}+\begin{bmatrix} 1-\alpha_{1t} \\ 1-\alpha_{2t}\end{bmatrix} \odot \begin{bmatrix} \epsilon_{1t} \\ \epsilon_{2t}\end{bmatrix}\\
\epsilon_{1t} &\sim \boldsymbol{p}^{(1)}_{\epsilon}=( p^{(1)}_{\epsilon_{1}},p^{(1)}_{\epsilon_{2}},p^{(1)}_{\epsilon_{3}},p^{(1)}_{\epsilon_{4}})\\ 
\epsilon_{2t} &\sim\boldsymbol{p}^{(2)}_{\epsilon}=( p^{(2)}_{\epsilon_{1}},p^{(2)}_{\epsilon_{2}},p^{(2)}_{\epsilon_{3}})\\
P(\epsilon_{1t},\epsilon_{2t})&=C_F(F_{\epsilon_{1t}}(\epsilon_{1t}),F_{\epsilon_{2t}}(\epsilon_{2t});\delta_{\epsilon})-C_F(F_{\epsilon_{1t}}(\epsilon_{1t}-1),F_{\epsilon_{2t}}(\epsilon_{2t});\delta_{\epsilon})\\
    &-C_F(F_{\epsilon_{1t}}(\epsilon_{1t}),F_{\epsilon_{2t}}(\epsilon_{2t}-1);\delta_{\epsilon})+
    C_F(F_{\epsilon_{1t}}(\epsilon_{1t}-1),F_{\epsilon_{2t}}(\epsilon_{2t}-1);\delta_{\epsilon})\\ 
\alpha_{1t} &\sim Bernoulli(\phi_1)\\ 
\alpha_{2t} &\sim Bernoulli(\phi_2)
    \end{align*}

    \item Model 4: Time series with different dependent random mechanisms, but independent innovations 
        \begin{align*}
\begin{bmatrix} Z_{1t} \\ Z_{2t}\end{bmatrix}&=\begin{bmatrix} \alpha_{1t} \\ \alpha_{2t}\end{bmatrix} \odot \begin{bmatrix} Z_{1,t-1} \\ Z_{2,t-1}\end{bmatrix}+\begin{bmatrix} 1-\alpha_{1t} \\ 1-\alpha_{2t}\end{bmatrix} \odot \begin{bmatrix} \epsilon_{1t} \\ \epsilon_{2t}\end{bmatrix}\\
\epsilon_{1t} &\sim \boldsymbol{p}^{(1)}_{\epsilon}=( p^{(1)}_{\epsilon_{1}},p^{(1)}_{\epsilon_{2}},p^{(1)}_{\epsilon_{3}},p^{(1)}_{\epsilon_{4}})\\ 
\epsilon_{2t} &\sim\boldsymbol{p}^{(2)}_{\epsilon}=( p^{(2)}_{\epsilon_{1}},p^{(2)}_{\epsilon_{2}},p^{(2)}_{\epsilon_{3}})\\
\alpha_{1t} &\sim Bernoulli(\phi_1)\\ 
\alpha_{2t} &\sim Bernoulli(\phi_2)\\
P(\alpha_{1t},\alpha_{2t})&=C_F(F_{\alpha_{1t}}(\alpha_{1t}),F_{\alpha_{2t}}(\alpha_{2t});\delta_{\alpha})-C_F(F_{\alpha_{1t}}(\alpha_{1t}-1),F_{\alpha_{2t}}(\alpha_{2t});\delta_{\alpha})\\
    &-C_F(F_{\alpha_{1t}}(\alpha_{1t}),F_{\alpha_{2t}}(\alpha_{2t}-1);\delta_{\alpha})+
    C_F(F_{\alpha_{1t}}(\alpha_{1t}-1),F_{\alpha_{2t}}(\alpha_{2t}-1);\delta_{\alpha}) 
    \end{align*}

    \item Model 5: Time series with different dependent random mechanisms and dependent innovations  
\begin{align*}
\begin{bmatrix} Z_{1t} \\ Z_{2t}\end{bmatrix}&=\begin{bmatrix} \alpha_{1t} \\ \alpha_{2t}\end{bmatrix} \odot \begin{bmatrix} Z_{1,t-1} \\ Z_{2,t-1}\end{bmatrix}+\begin{bmatrix} 1-\alpha_{1t} \\ 1-\alpha_{2t}\end{bmatrix} \odot \begin{bmatrix} \epsilon_{1t} \\ \epsilon_{2t}\end{bmatrix}\\
\epsilon_{1t} &\sim \boldsymbol{p}^{(1)}_{\epsilon}=( p^{(1)}_{\epsilon_{1}},p^{(1)}_{\epsilon_{2}},p^{(1)}_{\epsilon_{3}},p^{(1)}_{\epsilon_{4}})\\ 
\epsilon_{2t} &\sim\boldsymbol{p}^{(2)}_{\epsilon}=( p^{(2)}_{\epsilon_{1}},p^{(2)}_{\epsilon_{2}},p^{(2)}_{\epsilon_{3}})\\
P(\epsilon_{1t},\epsilon_{2t})&=C_F(F_{\epsilon_{1t}}(\epsilon_{1t}),F_{\epsilon_{2t}}(\epsilon_{2t});\delta_{\epsilon})-C_F(F_{\epsilon_{1t}}(\epsilon_{1t}-1),F_{\epsilon_{2t}}(\epsilon_{2t});\delta_{\epsilon})\\
    &-C_F(F_{\epsilon_{1t}}(\epsilon_{1t}),F_{\epsilon_{2t}}(\epsilon_{2t}-1);\delta_{\epsilon})+
    C_F(F_{\epsilon_{1t}}(\epsilon_{1t}-1),F_{\epsilon_{2t}}(\epsilon_{2t}-1);\delta_{\epsilon})\\ 
\alpha_{1t} &\sim Bernoulli(\phi_1)\\ 
\alpha_{2t} &\sim Bernoulli(\phi_2)\\
    P(\alpha_{1t},\alpha_{2t})&=C_F(F_{\alpha_{1t}}(\alpha_{1t}),F_{\alpha_{2t}}(\alpha_{2t});\delta_{\alpha})-C_F(F_{\alpha_{1t}}(\alpha_{1t}-1),F_{\alpha_{2t}}(\alpha_{2t});\delta_{\alpha})\\
    &-C_F(F_{\alpha_{1t}}(\alpha_{1t}),F_{\alpha_{2t}}(\alpha_{2t}-1);\delta_{\alpha})+
    C_F(F_{\alpha_{1t}}(\alpha_{1t}-1),F_{\alpha_{2t}}(\alpha_{2t}-1);\delta_{\alpha})  
\end{align*}
  
\end{itemize}
In all cases $C_F(.;\delta)$ denotes the Frank copula described in Equation \ref{Equation: Frank copula}, while $p^{(1)}_{\epsilon_{4}}=1-\sum_{j}^{3}p^{(1)}_{\epsilon_{j}}$ and $p^{(2)}_{\epsilon_{3}}=1-\sum_{j}^{2}p^{(2)}_{\epsilon_{j}}$.

Fitting the above models, the results are presented in Table \ref{Tab. Application Estimates2}. Standard errors are obtained from the Hessian. Based on the estimates of $\phi_1$ and $\phi_2$ of Models 1,3,4 and 5, we can see that random mechanism is similar, thus it is useful trying a more parsimonious model that assumes one common random mechanism for both time series. This is Model 2.

\begin{table}[t]
\caption{\label{Tab. Application Estimates2} Estimates with standard errors, log-likelihood, number of parameters and information criteria under different models}
\centering
\begin{tabular}{cccccc} 
 \hline
          & Model 1 & Model 2  & Model 3 & Model 4 & Model 5 \\  
 \hline
 $\hat{p}^{(1)}_{\epsilon_{1}}$ & $\underset{(0.097)}{0.143}$ & $\underset{(0.089)}{0.133}$ & $\underset{(0.094)}{0.140}$ & $\underset{(0.096)}{0.141}$  & $\underset{(0.082)}{0.123}$ \\
 $\hat{p}^{(1)}_{\epsilon_{2}}$ & $\underset{(0.109)}{0.164 }$ & $\underset{(0.103)}{0.159}$ & $\underset{(0.107)}{0.165}$ & $\underset{(0.106)}{0.158}$ & $\underset{(0.094)}{0.145}$  \\
 $\hat{p}^{(1)}_{\epsilon_{3}}$ & $\underset{(0.136)}{0.270}$ &$\underset{(0.131)}{0.266}$ & $\underset{(0.133)}{0.271}$ & $\underset{(0.138)}{0.276}$   & $\underset{(0.126)}{0.303}$ \\
 $\hat{p}^{(2)}_{\epsilon_{1}}$ & $\underset{(0.141)}{0.316}$ & $\underset{(0.139)}{0.386}$ & $\underset{(0.139)}{0.393}$ & $\underset{(0.142)}{0.319}$  & $\underset{(0.11)}{0.410}$ \\ 
 $\hat{p}^{(2)}_{\epsilon_{2}}$ & $\underset{(0.142)}{ 0.467}$ & $\underset{(0.139)}{0.431}$ & $\underset{(0.138)}{0.428}$ & $\underset{(0.142)}{0.471}$  & $\underset{(0.11)}{0.433}$ \\ 
 $\hat{\phi}_1$                 & $\underset{(0.045)}{0.857}$ & $\underset{(0.036)}{0.830}$ & $\underset{(0.047)}{0.851}$ & $\underset{(0.045)}{0.855}$  & $\underset{(0.051)}{0.823}$ \\ 
 $\hat{\phi}_2$                 & $\underset{(0.048)}{0.824}$ & - & $\underset{(0.052)}{0.809}$ & $\underset{(0.049)}{0.823}$  & $\underset{(0.054)}{0.752}$ \\
 $\hat{\delta}_{\epsilon}$      & $-$                         & $\underset{(0.566)}{26.730}$ & $\underset{(0.416)}{27.335}$ & -                   & $\underset{(0.156)}{28.4}$ \\ 
 $\hat{\delta}_{\alpha}$        & $-$                         & $-$                         & -                 & $\underset{(2.807)}{2.367}$   & $\underset{(0.687)}{24.719}$ \\ 
 \hline
 Log-lik.  &  -87.65 & -85.86 & -85.67 & -87.20 & -82.22 \\
\hline 
\verb|#|parameters & 7 & 7 & 8 & 8 & 9\\
 \hline
 BIC & 207.75 & \textbf{204.15} & 208.42 & 211.47 & 206.16 \\
 \hline
 AIC & 189.31 & 185.71 & 187.34 & 190.40 & \textbf{182.44}\\
 \hline
\end{tabular}

\end{table}

To decide which is the optimal model we use information criteria (BIC, AIC) because not all models are nested. According to BIC the best model is Model 2 while based on AIC the best models is Model 5. The two models are nested, thus we can compare them based on Likelihood Ratio Test (LRT). The results showed that there is statistically significant difference between the two models (LRT: pvalue$=0.026 < 5\%$). Thus, we finally choose Model 5:
\begin{align*}
\begin{bmatrix} Z_{1t} \\ Z_{2t}\end{bmatrix}&=\begin{bmatrix} \alpha_{1t} \\ \alpha_{2t}\end{bmatrix} \odot \begin{bmatrix} Z_{1,t-1} \\ Z_{2,t-1}\end{bmatrix}+\begin{bmatrix} 1-\alpha_{1t} \\ 1-\alpha_{2t}\end{bmatrix} \odot \begin{bmatrix} \epsilon_{1t} \\ \epsilon_{2t}\end{bmatrix}\\
\epsilon_{1t} &\sim \hat{\boldsymbol{p}}^{(1)}_{\epsilon}=( \hat{p}^{(1)}_{\epsilon_{1}}=0.123,\hat{p}^{(1)}_{\epsilon_{2}}=0.145,\hat{p}^{(1)}_{\epsilon_{3}}=0.303,\hat{p}^{(1)}_{\epsilon_{4}}=0.429)\\ 
\epsilon_{2t} &\sim\hat{\boldsymbol{p}}^{(2)}_{\epsilon}=( \hat{p}^{(2)}_{\epsilon_{1}}=0.410,\hat{p}^{(2)}_{\epsilon_{2}}=0.433,\hat{p}^{(2)}_{\epsilon_{3}}=0.157)\\
\alpha_{1t} &\sim \text{Bernoulli}(\hat{\phi}_1=0.823)\\ 
\alpha_{2t} &\sim \text{Bernoulli}(\hat{\phi}_2=0.752) 
\end{align*}
where the joint distribution of $\boldsymbol{\alpha}_t$ is given by a Frank Copula with $\hat{\delta}_{\alpha}=28.4$ and
the joint distribution of $\boldsymbol{\epsilon}_t$ is also given by a Frank Copula with $\hat{\delta}_{\epsilon}=24.719$. The joint pmfs are presented in Figure \ref{Fig:Joint pmf}.

\begin{figure}[t]
\begin{center}
\includegraphics[scale=0.35]{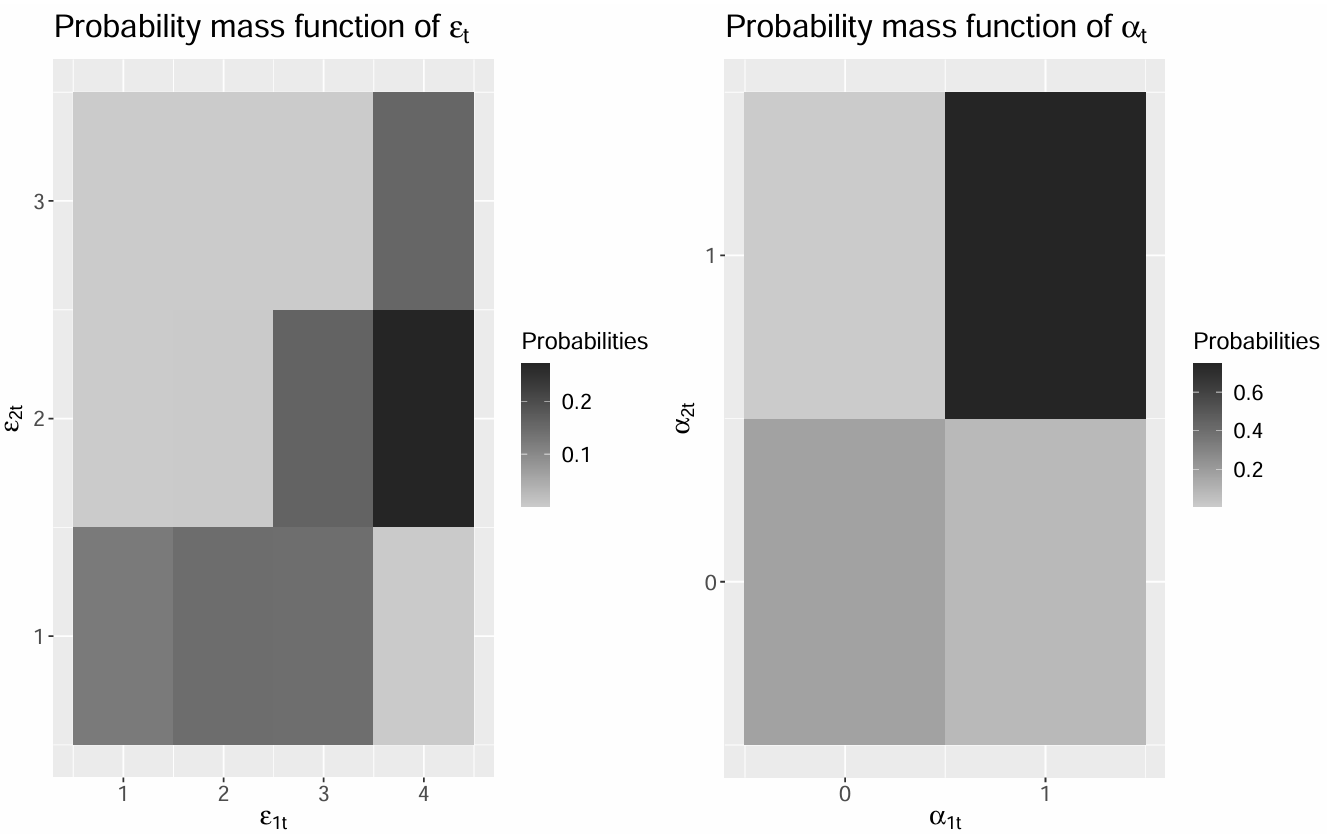}\caption{\label{Fig:Joint pmf} Estimated joint probability mass functions of innovation terms and random mechanisms}
\end{center}
\end{figure}

As we expect from the high values of autocorrelation, the probability of choosing the previous state $(\phi_1,\phi_2)$ for the current value is very high for both time series, indicating that the series are persistent to the same state. This can also be seen in Figure \ref{Fig:Application}, where we can see that there are long periods for which both time series, especially Slovakia, that presents higher probability, don't change state. In addition, we can conclude that Slovakia has bigger problem of unemployment than Czechia, while it has higher probabilities in the highest state. On the other hand, compared to Slovakia, Czechia presents highest probabilities in the two lowest unemployment states.

\subsection{Forecasting}
Assume the case that we would like to use the model for forecasting the next states of unemployment per each quarter based on the model we have chosen. More specifically, we assume that for both countries we have available only the unemployment states for the period $1998-2020$ ($t=1,\ldots,92$), and we would like to predict the unemployment state for the period $2021-2023$. Based on data until $2020$ we fit the chosen model. Then, to forecast $\boldsymbol{Z}_{t+h}$ for $h=1,\ldots,12$, first we define the marginal probabilities of each time series, conditional to the their previous values:  
\begin{align*}
 P(Z_{1,t+h}=i|\boldsymbol{Z}_{t+h-1})=\sum_{j=1}^{3}P(Z_{1,t+h}=i,Z_{2,t+h}=j|\boldsymbol{Z}_{t+h-1})  \\
 P(Z_{2,t+h}=i|\boldsymbol{Z}_{t+h-1})=\sum_{i=1}^{4}P(Z_{1,t+h}=i,Z_{2,t+h}=j|\boldsymbol{Z}_{t+h-1}).  
\end{align*}
Then, the forecast $\hat{Z}_{1,t+h}$ will occur as a random draw from the distribution $Z_{1,t+h}|\boldsymbol{Z}_{t+h-1}$. For $h=1$, $\boldsymbol{Z}_{t+h-1}$ are the final observed values of the time series ($\boldsymbol{Z}_{92}$), while for $h>1$, we use the last forecasts of the series $\boldsymbol{Z}_{t+h-1}$=$\hat{\boldsymbol{Z}}_{t+h-1}$. We repeat this procedure for $B=10000$ times. As final forecast at $h-$step we consider the most frequent state of the $B=10000$ repetitions. We follow the same procedure for forecasting the other series. 
We also consider the joint distribution of $Z_{1,t+h}$ and $Z_{2,t+h}$ for each $h=1,\ldots,12$, based on the 10000 simulations. The joint probability mass function is presented in Figure \ref{Fig:Forc_joint}. 

The forecasted values are presented in Tables \ref{Tab. Probs} and \ref{Tab. Forecasts}. Table \ref{Tab. Probs}
presents the marginal relative frequencies for the two forecasts based on the simulations, while Table \ref{Tab. Forecasts} shows the modal forecasts for each margin separately but also jointly. 

\begin{table}[t]
\caption{\label{Tab. Probs} Relative frequency table of states based $10000$ simulations for each time series }
\centering
\begin{tabular}{|c|cccc|ccc|} 
\hline 
  & & $Z_{1,t+h}$ & &  & & $Z_{2,t+h}$  & \\
 \hline
    $h$ & 1 & 2 & 3 & 4 & 1 & 2 & 3  \\           
 \hline
1& 0.011 & 0.830 & 0.064 & 0.095 & 0.829 & 0.125 & 0.046\\
2& 0.022 & 0.696 & 0.114 & 0.167 & 0.703 & 0.216 & 0.082\\
3& 0.030 & 0.585 & 0.155 & 0.231 & 0.612 & 0.279 & 0.109\\
4& 0.037 & 0.500 & 0.190 & 0.273 & 0.544 & 0.329 & 0.127\\
5& 0.043 & 0.426 & 0.221 & 0.309 & 0.492 & 0.370 & 0.138\\
6& 0.046 & 0.369 & 0.240 & 0.345 & 0.460 & 0.396 & 0.144\\
7& 0.051 & 0.322 & 0.256 & 0.371 & 0.435 & 0.418 & 0.147 \\
8& 0.053 & 0.284 & 0.269 & 0.394 & 0.416 & 0.435 & 0.150 \\
9  & 0.057 & 0.255 & 0.280 & 0.408 & 0.399 & 0.443  & 0.158\\
10 & 0.057 & 0.232 & 0.291 & 0.420 & 0.392 & 0.444 & 0.164\\
11 & 0.058 & 0.210 & 0.302 & 0.430 & 0.386 & 0.448 & 0.166\\
12 & 0.062 & 0.193 & 0.307 & 0.438 & 0.382 & 0.455 &0.163\\
 \hline
\end{tabular}

\end{table}

\begin{table}[t]
\caption{\label{Tab. Forecasts} Forecasts based on $10000$ simulations and true values of the time series}
\centering
\begin{tabular}{|c|cc|cc|c|} 
 \hline
    $h$ &  $\hat{Z}_{1,t+h}$\rule{0pt}{3ex}  & True value &  $\hat{Z}_{2,t+h}$  & True value & ($\hat{Z}_{1,t+h}$,$\hat{Z}_{2,t+h}$)\\           
 \hline
1 & 2 & 2& 1 & 1 &  (2,1)\\
2 & 2 & 2& 1 & 1 &  (2,1)\\
3 & 2 & 2& 1 & 1 &  (2,1)\\
4 & 2 & 2& 1 & 1 &  (2,1)\\
5 & 2 & 2& 1 & 1 &  (2,1)\\
6 & 2 & 2& 1 & 1 &  (2,1)\\
7 & 4 & 2& 1 & 1 & (4,2)\\
8 & 4 & 2& 2 & 1 & (4,2)\\
9  & 4 & 2& 2 & 1 & (4,2)\\
10 & 4 & 1& 2 & 1 & (4,2)\\
11 & 4 & 1& 2 & 1 & (4,2)\\
12 & 4 & 1& 2 & 1 & (4,2)\\
 \hline
\end{tabular}
\end{table}

\begin{figure}[t]
\begin{center}
\includegraphics[scale=0.35]{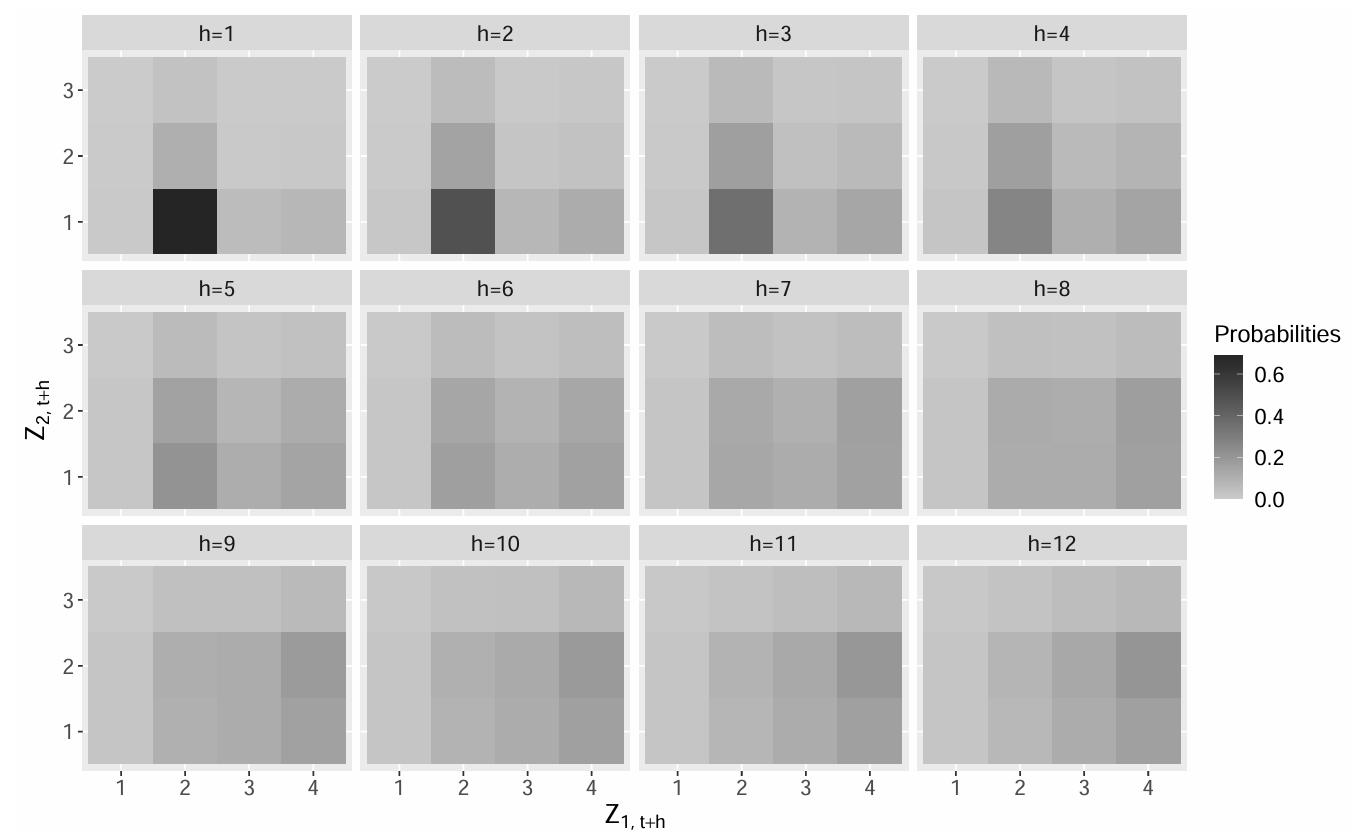}\caption{\label{Fig:Forc_joint} Joint probability mass function of the two series for each step ahead based on 10000 simulations}
\end{center}
\end{figure}

\section{Conclusion}
\label{Sec7}
 The paper introduces the BDAR(1) model, an extension of the well-known DAR(1) model, designed for bivariate discrete-valued time series such as integer, ordinal and binary data. The bivariate DAR(1) model must capture two types of dependence: the serial correlation within each individual series and the cross-correlation between the two series. Serial correlation is modeled through a Bernoulli variable for each series, which determines whether the current state is inherited from the previous state or drawn from an innovation term. Dependence between the two processes is incorporated via both the Bernoulli vectors and the innovation terms. Specifically, the model assumes that not only the random mixtures of the two series are correlated, but also their innovations. To capture this, the joint distributions are specified using copulas, which allow for flexible modeling of joint distributions of discrete variables and can represent a wide range of dependence structures. Estimation of the model is carried out by maximizing the conditional log-likelihood, and the paper also discusses the basic properties of the model.

 Special attention is given to the case of two ordinal time series, as the existing literature on bivariate models for this setting is quite limited. To address this gap, a simulation study is conducted to evaluate the model’s performance across different sample sizes. The results indicate that the model requires at least a moderate sample size to produce robust estimates. Furthermore, the proposed methodology was applied to jointly model and forecast the unemployment states of Slovakia and the Czech Republic.

The introduction of the BDAR($1$) model opens up several promising directions for further research. First, the proposed methodology accounts only for serial dependence of lag $1$. However, in many applications, serial dependence at higher lags is also relevant. Incorporating additional lagged values would require representing the random mechanisms with multinomial rather than Bernoulli vectors, thereby increasing the number of possible outcomes. This, in turn, would result in a more complex model with a larger set of parameters. Another possible extension involves moving to higher dimensions, as there may be cases where a joint model for more than two time series is needed.

Since the BDAR($1$) model is an extension of the DAR($1$) model, it inherits some of its properties. One such property is that the model captures only non-negative serial dependence. In the special case of binary time series, defining and modeling negative serial dependence is relatively straightforward \citep{jentsch2019generalized}. However, for ordinal time series, this task is far from trivial. This challenge persists even in the univariate case.

The proposed model is sufficiently general to be applied not only to ordinal time series but also to count and binary time series. We also expect that continuous time series could be incorporated, provided an extension similar to that of \cite{moller2020generalized} is adopted. Such an extension could also help address the issue of long runs of repeating values, as previously discussed. Moreover, since copulas make it possible to define the joint distribution for data of any type, the proposed model could be employed for time series of different modalities. Finally, in the case of mixed data—or even for time series of the same type but with different ranges—the inclusion of cross-correlation terms poses a non-trivial challenge, which represents another interesting avenue for future research.

\section*{Acknowledgment}
The research work was supported by the Hellenic Foundation 
for Research and Innovation (HFRI) under the 5th Call for 
HFRI PhD Fellowships (Fellowship Number: 20535.)

\bibliographystyle{apalike}
\bibliography{biblio}
\end{document}